\documentclass[final]{siamltex}

\usepackage{amssymb}
\usepackage{epsfig}
\usepackage{amsmath}
\usepackage{amsfonts}

\title{Dilatation of a one-dimensional nonlinear crack impacted by a periodic elastic wave\thanks{This work was carried out with the support of the GdR 2501, CNRS, France.}}

\author{St\'ephane Junca\thanks{Laboratoire J.A. Dieudonn\'e, UMR 6621 CNRS, Universit\'e de Nice Sophia-Antipolis, Parc Valrose, 06108 Nice Cedex 02, France ({\tt junca@math.unice.fr}).}
\and Bruno Lombard\thanks{Laboratoire de M\'ecanique et d'Acoustique, UPR 7051 CNRS, 31 chemin Joseph Aiguier, 13402 Marseille, France ({\tt lombard@lma.cnrs-mrs.fr}).}}

\begin{document}

\maketitle

\begin{abstract}
The interactions between linear elastic waves and a nonlinear crack with finite compressibility are studied in the one-dimensional context. Numerical studies on a hyperbolic model of contact with sinusoidal forcing have shown that the mean values of the scattered elastic displacements are discontinuous across the crack. The mean dilatation of the crack also increases with the amplitude of the forcing levels. The aim of the present theoretical study is to analyse these nonlinear processes under a larger range of nonlinear jump conditions. For this purpose, the problem is reduced to a nonlinear differential equation. The dependence of the periodic solution on the forcing amplitude is quantified under sinusoidal forcing conditions. Bounds for the mean, maximum and minimum values of the solution are presented. Lastly, periodic forcing with a null mean value is addressed. In that case, a result about the mean dilatation of the crack is obtained.
\end{abstract}



\pagestyle{myheadings}
\thispagestyle{plain}
\markboth{S. JUNCA AND B. LOMBARD}{NONLINEAR CRACK AND ELASTIC WAVES}

\section{Introduction}\label{SecIntro}	

\subsection{Aims}

The modeling of interactions between ultrasonic waves and cracks is of great interest in many fields of applied mechanics. When the wavelengths are much larger than the width of the cracks, the latter are usually replaced by zero-thickness interfaces with appropriate jump conditions. Linear models for crack-face interactions have been widely developed \cite{Pyrak90}. However, these models do not prevent the non-physical penetration of crack faces. In addition, laboratory experiments have shown that ultrasonic methods based on linear models often fail to detect partially closed cracks \cite{Solodov98}.

A well-known nonlinear model for cracks is the unilateral contact model \cite{Richardson79,JuRo08}. A more realistic hyperbolic model accounting for the finite compressibility of crack faces under normal loading conditions has been presented   for applications in engineering \cite{Achenbach82} and geomechanical contexts \cite{Bandis83}. The well-posedness and the numerical modeling of the latter model in the context of 1-D linear elastodynamics was previously studied in \cite{Lombard07}. The generation of harmonics and the distortion of simulated scattered velocity and stress waves were also observed. 

Subsequent numerical experiments have brought to light an interesting property. In the case of a sinusoidal incident wave, the mean values of the elastic displacements around a crack were found to be discontinuous. The simulations conducted also indicated that the mean dilatation of the crack increases with the amplitude of the forcing levels. This purely nonlinear phenomenon (the linear models predict no dilatation) is of physical interest: experimenters can measure the dilatation \cite{Korshak02} and use the data obtained to deduce the properties of the crack. 

A preliminary theoretical study based on the use of a perturbation method was presented in \cite{Lombard08}. An analytical expression for the dilatation was obtained in that study and successfully compared with the numerical results. However, this expression gives only a local estimate, valid in the particular case of the hyperbolic model subjected to very low sinusoidal forcing levels. The aim of the present study is to provide a theoretical analysis which can be applied to any range of forcing levels and to a larger number of contact models.

%-------------------------------------------------------------------------------------

\subsection{Design of the study}

The present paper is organized as follows:
\begin{itemize}
\item In section 2, the physical configuration is described. The family of strictly increasing convex jump conditions dealt with in this study is introduced. Numerical simulations of the process of interest are presented;
\item In section 3, the problem is reduced to a nonautonomous differential equation. Up to section 6, only the case of sinusoidal forcing, which allows a complete understanding of the phenomena involved, is studied;
\item In section 4, preliminary results on the differential equation are presented. Classical tools for dynamical systems are used: Poincar\'e map, phase portrait, lower and upper solutions \cite{Hubbard91,Nayfeh95};
\item In section 5, the main qualitative results of this study are presented. The mean, maximum and minimum aperture of the crack are bounded, and local estimates for small forcing are also proved;
\item In section 6, some of the previous results are extended to non-monochromatic periodic forcing conditions. The increase of the mean dilatation with the forcing parameter is proved in theorem \ref{TheoremYbarGene}. The most general quantitative result of the paper is given in equation (\ref{YbarLocalGene});
\item In section 7, some conclusions are drawn about the physical observables. Some future perspectives are also suggested.
\end{itemize}

%-------------------------------------------------------------------------------------
%-------------------------------------------------------------------------------------

\section{Statement of the problem}\label{SecPbStatement}

\subsection{Physical modeling}\label{SecModeling}

\begin{figure}[htbp]
\begin{center}
\includegraphics[scale=0.7]{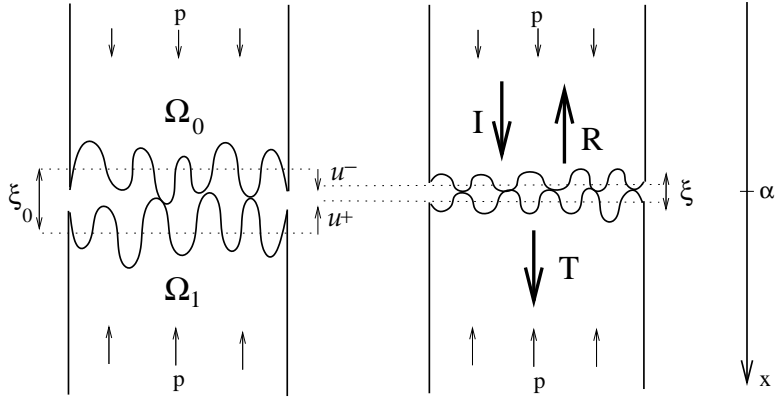}
\caption{Elastic media $\Omega_0$ and $\Omega_1$ with rough contact surfaces, under constant static stress $p$. Static (left) and dynamic (right) case, with incident (I), reflected (R) and transmitted (T) waves.}
\label{FigCrack}
\end{center}
\end{figure}

We consider the case of a single crack with rough faces separating two media $\Omega_0$ and $\Omega_1$, which are both linearly elastic and isotropic, taking $\rho$ to denote the density and $c$ to denote the elastic speed of the compressional waves. These parameters are piecewise constant and may be discontinuous around the crack: $(\rho_0,\,c_0)$ if $x \in \Omega_0$, $(\rho_1,\,c_1)$ if $x \in \Omega_1.$ The media are subject to a constant static stress $p$. At rest, the distance between the planes of average height is $\xi_0(p)>0$ (figure \ref{FigCrack}, left). 

Elastic compressional waves are emitted by a singular source of stress at $x=x_s<\alpha$ in $\Omega_0$, where $\alpha$ is a median plane of the actual flaw surface. The wave impacting $\alpha$ gives rise to reflected (in $\Omega_0$) and transmitted (in $\Omega_1$) compressional waves. These perturbations in $\Omega_0$ and $\Omega_1$ are described by the 1-D elastodynamic equations 
\begin{equation}
\rho \,\frac{\textstyle \partial \,v}{\textstyle \partial \,t}=\frac{\textstyle \partial\,\sigma}{\textstyle \,\partial\, x},\qquad
\frac{\textstyle \partial \,\sigma}{\textstyle \partial \,t}=
\rho \,c^2 \,\frac{\textstyle \partial \,v}{\textstyle \partial \,x}+S(t)\,\delta(x-x_s),
\label{LCscal}
\end{equation}
where $S(t)$ denotes the causal stress source, $v=\frac{\partial\,u}{\partial\,t}$ is the elastic velocity, $u$ is the elastic displacement, and $\sigma$ is the elastic stress perturbation around $p$. The dynamic stresses induced by the elastic waves affect the thickness $\xi(t)$ of the crack (figure \ref{FigCrack}, right). The constraint
\begin{equation}
\xi=\xi_0+[u]\geq \xi_0-d>0
\label{Penetration}
\end{equation}
must be satisfied, where $[u]=u^+-u^-$ is the difference between the elastic displacements on the two sides of the crack, and $d(p)>0$ is the {\it maximum allowable closure} \cite{Bandis83}. We also assume that the wavelengths are much larger than $\xi$, so that the propagation time across the crack is neglected, the latter being replaced by a zero-thickness {\it interface} at $x=\alpha$: $[u]=[u(\alpha,\,t)]=u(\alpha^+,\,t)-u(\alpha^-,\,t)$. 

\begin{figure}[htbp]
\begin{center}
\begin{tabular}{cc}
%(a) & (b)\\
\includegraphics[scale=0.65]{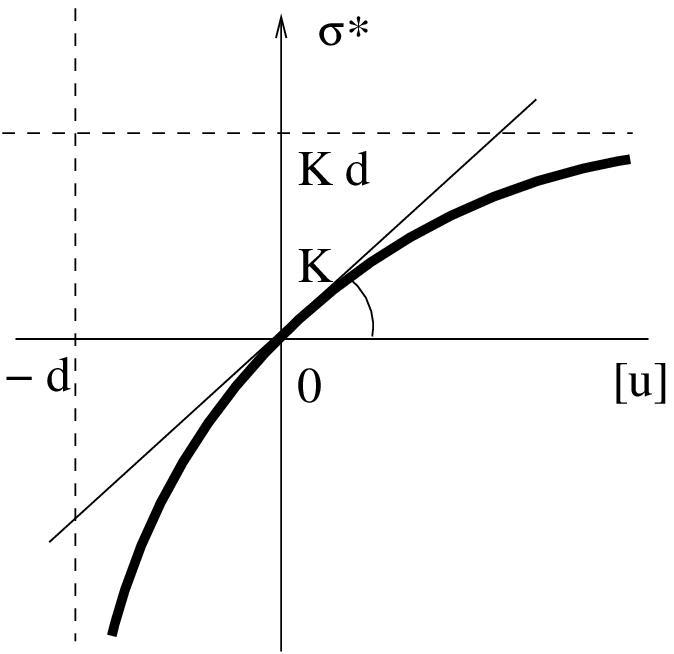} &
\includegraphics[scale=0.65]{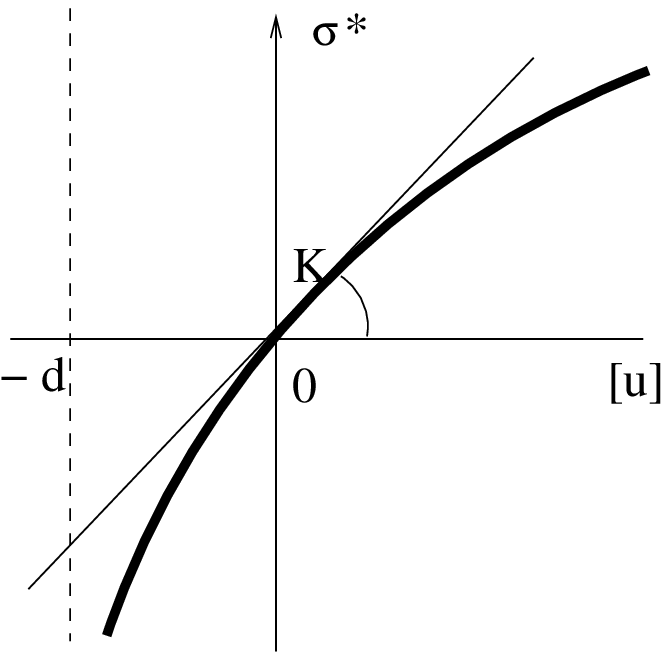} 
\end{tabular}
\end{center}
\caption{Sketch of the nonlinear relation between the stress and the jump of the elastic displacement at $\alpha$. Left row: model 1 (\ref{Model1}), right row: model 2 (\ref{Model2}).}
\label{FigContact}
\end{figure}

Two independent jump conditions at $\alpha$ need to be defined to obtain a well-posed problem. The discontinuity of $\sigma$ is proportional to the mass of the interstitial medium present between $\Omega_0$ and $\Omega_1$ \cite{Rokhlin91}. Since the crack is dry and contains only air, the density of which is much smaller than $\rho_0$ or $\rho_1$, the elastic stress is assumed to be continuous:
\begin{equation}
\left[\sigma(\alpha,\,t)\right]= 0\,\Rightarrow\,\sigma(\alpha^+,\,t)=\sigma(\alpha^-,\,t)=\sigma^*(t).
\label{JCsigma}
\end{equation}
Establishing the second jump condition is a more complex task. Experimental and theoretical studies have shown that $u$ is discontinuous, and that the discontinuity is proportional to the stress applied. The linear model has often been considered \cite{Pyrak90}:
\begin{equation}
\sigma^*(t)=K\,\left[u(\alpha,\,t)\right],
\label{JClin}
\end{equation}
where $K$ is the {\it interfacial stiffness}. Welded conditions $[u(\alpha,\,t)]=0$ are obtained if $K \rightarrow +\infty$. However, the linear condition (\ref{JClin}) violates (\ref{Penetration}) under large compression loadings conditions: $\sigma^*(t)<-K\,\delta \, \Rightarrow \, \xi <\xi_0-\delta$. The linear condition (\ref{JClin}) is therefore realistic only with very small perturbations. With larger ones, a nonlinear jump condition is required.

To develop this relation, it should be noted that compression loading increases the surface area of the contacting faces. A smaller stress is therefore needed to open than to close a crack; an infinite stress is even required to close the crack completely. In addition, the constraint (\ref{Penetration}) must be satisfied, and the model must comply with (\ref{JClin}) in the case of  small stresses. Lastly, concave stress-closure relations have been observed experimentally \cite{Malama03}. Dimensional analysis shows that the general relation
\begin{equation}
\sigma^*(t)=K\,d\,{\cal F}\left([u(\alpha,\,t)]/d\right)
\label{JCnonlin}
\end{equation}
is suitable, where ${\cal F}$ is a smooth increasing concave function
\begin{equation}
\begin{array}{l}
\displaystyle
{\cal F}:\,]-1,\,+\infty[\rightarrow]-\infty,\,{\cal F}_{\max}[,\quad \lim_{X\rightarrow-1}{\cal F}(X)=-\infty,\quad 0<{\cal F}_{\max}\leq +\infty,\\
[6pt]
\displaystyle
{\cal F}(0)=0,\quad {\cal F}^{'}(0)=1,\quad {\cal F}^{''}<0<{\cal F}^{'}.
\end{array}
\label{Fconcave}
\end{equation}
Two models illustrate the nonlinear relation (\ref{JCnonlin}). First, the so-called {\it model 1} presented in \cite{Achenbach82,Bandis83} is 
\begin{equation}
\sigma^*(t)=\frac{\textstyle K\,\left[u(\alpha,\,t)\right]}{\textstyle 1+\left[u(\alpha,\,t)\right]/d}\,\Leftrightarrow \,{\cal F}(X)=\frac{\textstyle X}{\textstyle 1+X},\quad {\cal F}_{\max}=1.
\label{Model1}
\end{equation}
Secondly, the so-called {\it model 2} presented in \cite{Malama03} is
\begin{equation}
\sigma^*(t)=K\,d\,\ln\left(1+[u(\alpha,\,t)]/d\right)\,\Leftrightarrow\, {\cal F}(X)=\ln(1+X),\quad {\cal F}_{\max}=+\infty.
\label{Model2}
\end{equation}
These two models are sketched in figure \ref{FigContact}. The straight line with a slope $K$ tangential to the curves at the origin gives the linear jump conditions (\ref{JClin}). 

%-------------------------------------------------------------------------------------

\subsection{Numerical experiments}\label{SecNumExp}

\begin{figure}[htbp]
\begin{center}
\begin{tabular}{cc}
(a) & (b)\\
\includegraphics[scale=0.31]{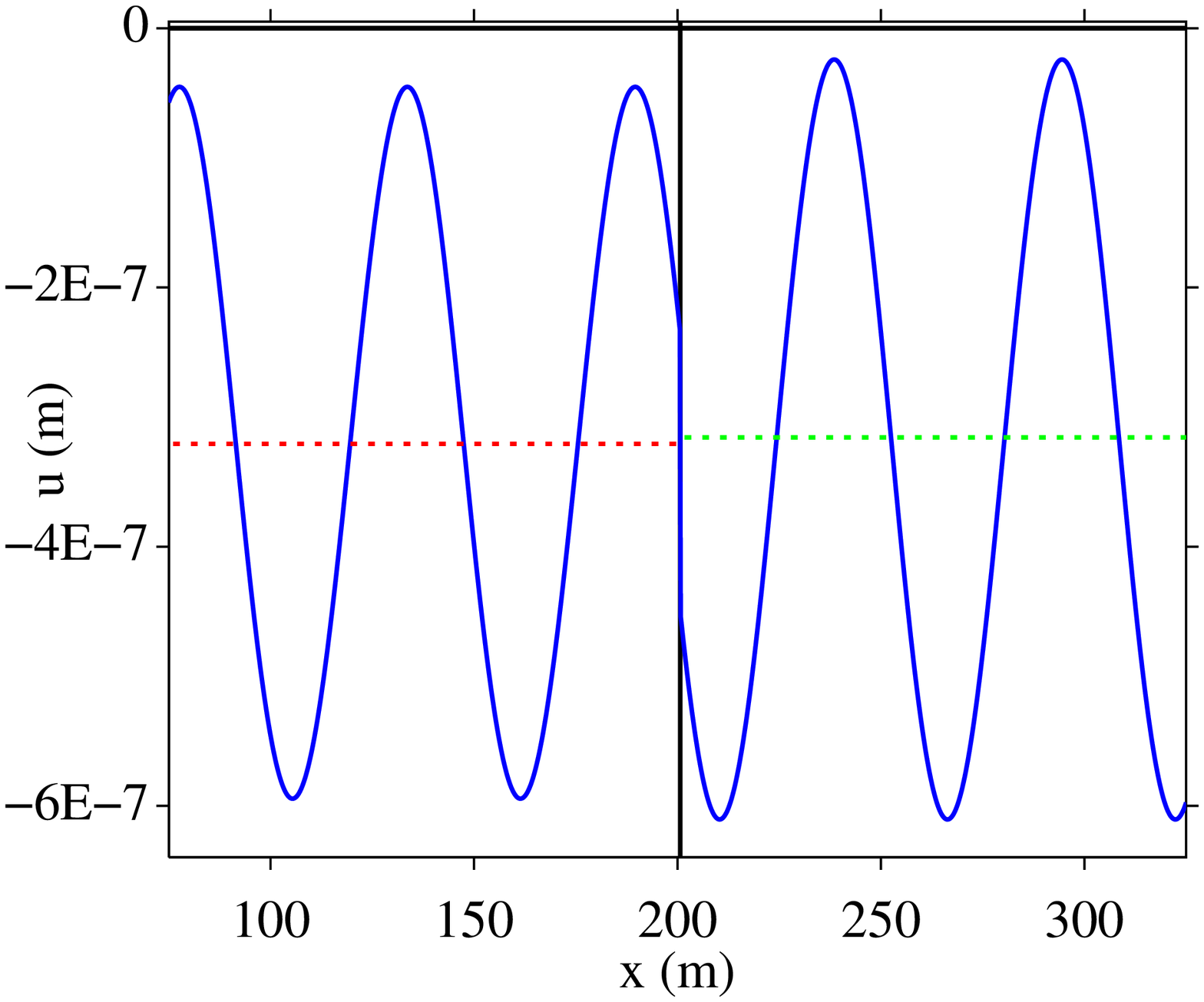}&
\includegraphics[scale=0.31]{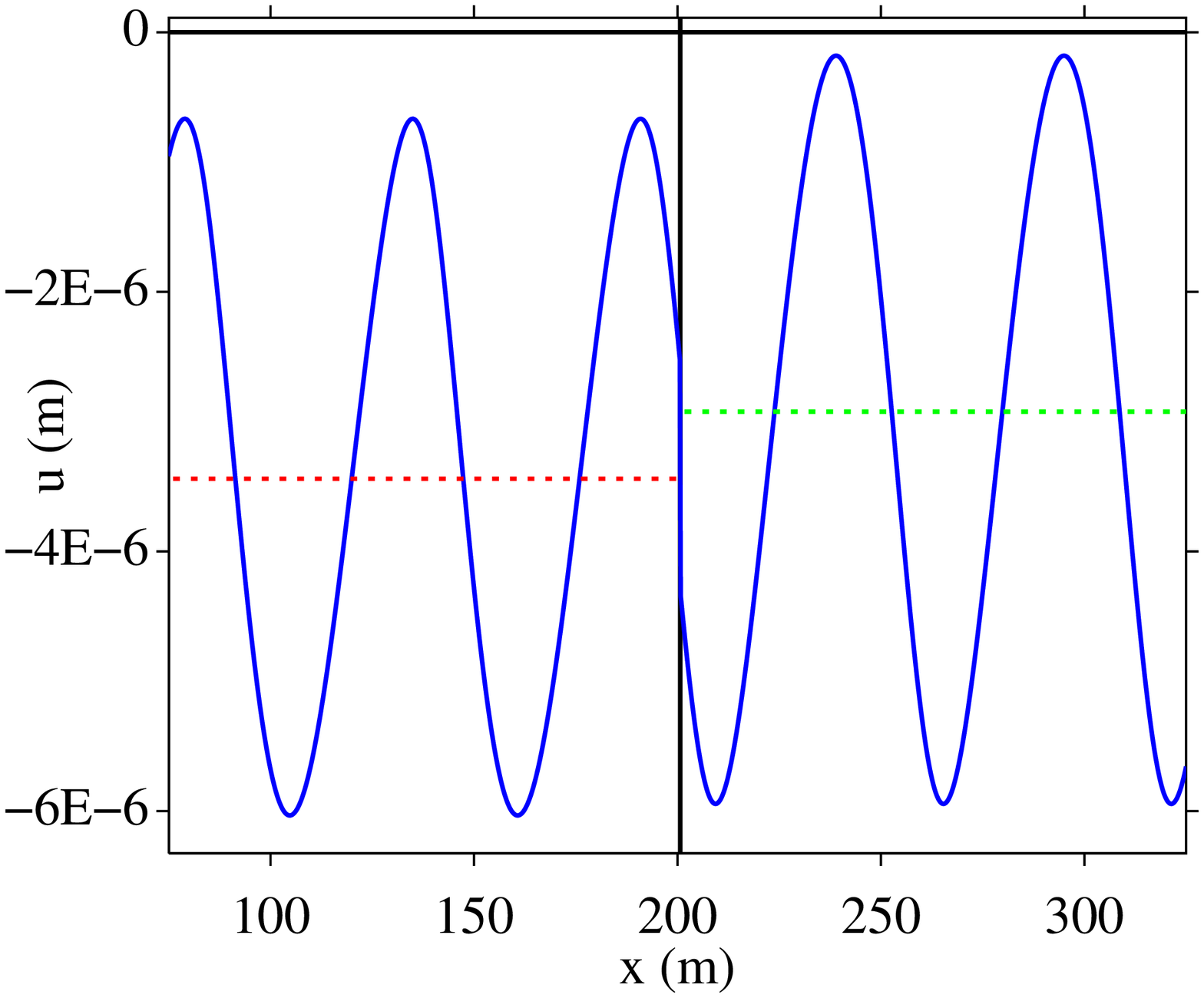}\\
(c) & (d)\\
\includegraphics[scale=0.31]{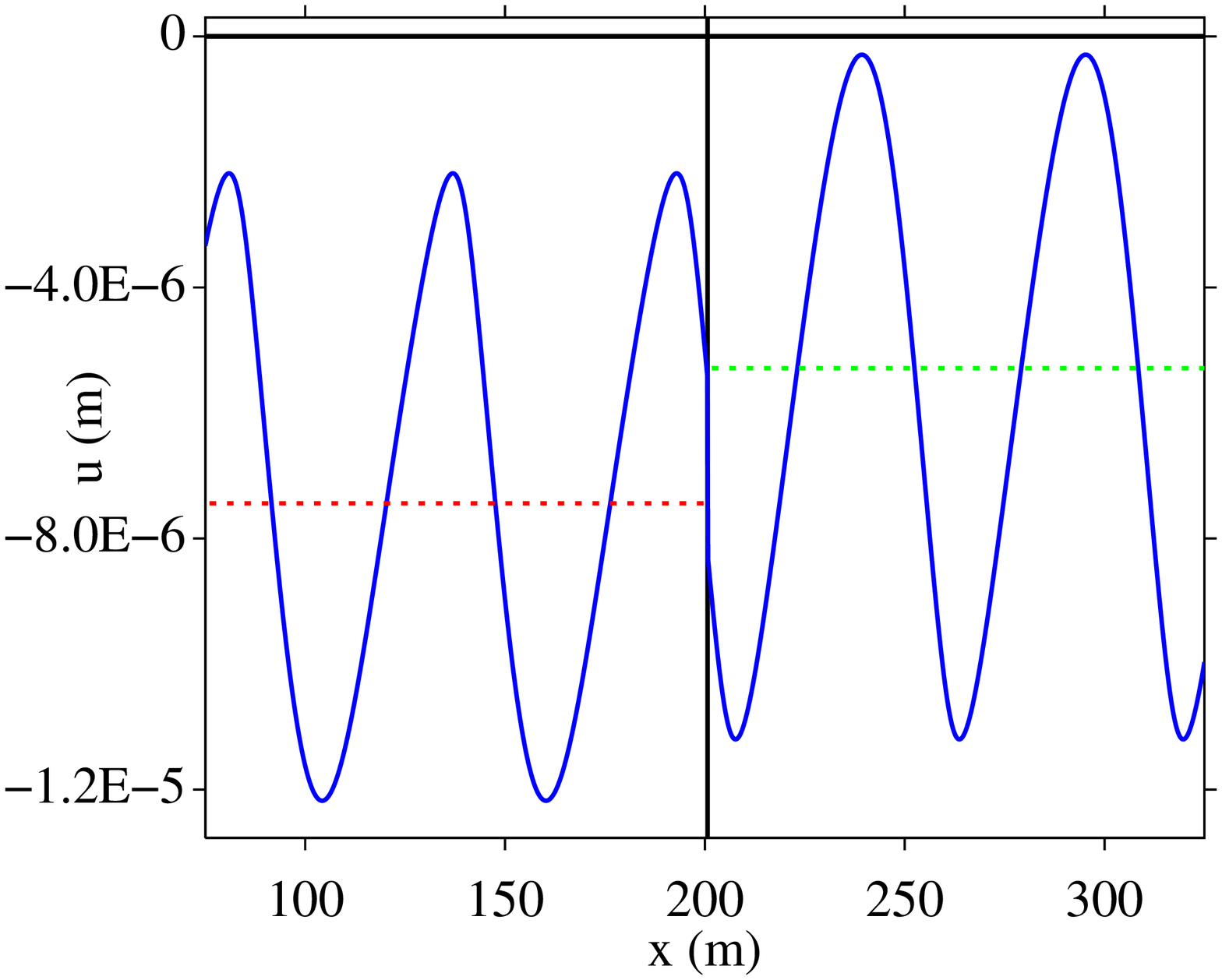}&
\includegraphics[scale=0.31]{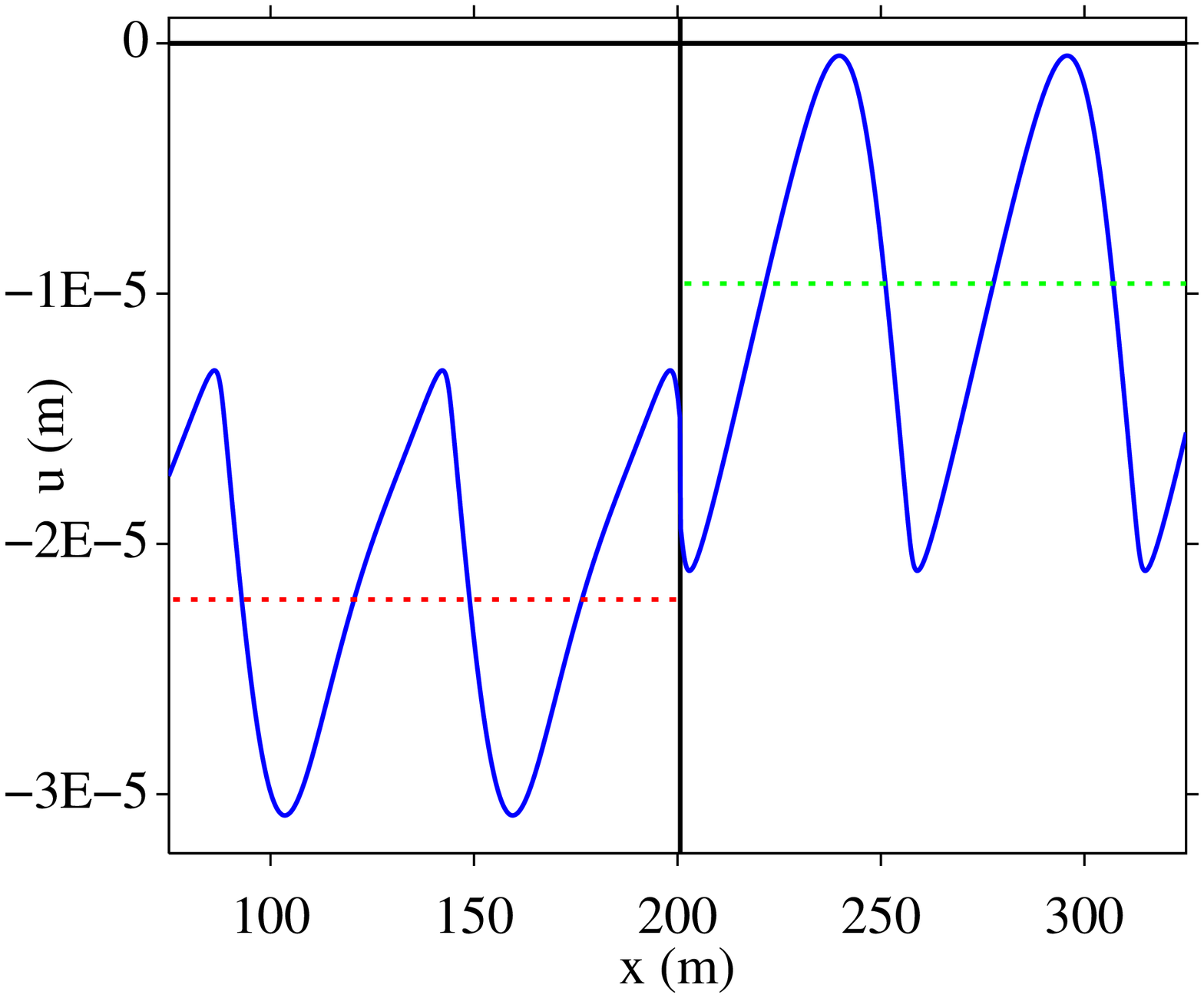}\\
\end{tabular}
\end{center}
\caption{Model 1 (\ref{Model1}): snapshots of the elastic displacement, for various amplitudes $v_0$ of the incident elastic velocity: $10^{-4}$ m/s (a), $10^{-3}$ m/s (b), $2\,10^{-3}$ m/s (c) and $5\,10^{-3}$ m/s (d). The vertical solid line denotes the location $\alpha=200$ m of the crack. The red and green dotted horizontal lines denote the mean value of the elastic displacement on both sides of $\alpha$.} 
\label{FigNumU}
\end{figure}

Here, we describe the influence of the nonlinear jump condition (\ref{JCnonlin}) on the wave scattering. For this purpose, we consider a single crack described by model 1 (\ref{Model1}), a sinusoidal source $S$ with a frequency of 50 Hz, and parameters
$$
\left\{
\begin{array}{l}
\rho_0=\rho_1=1200\,\mbox{ kg.m}^{-3},\quad K=1.3\,10^9\,\mbox{ kg.m}^{-1}\mbox{.s}^{-2},\\
[5pt]
c_0=c_1=2800\,\mbox{  m.s}^{-1}, \hspace{0.65cm} d=6.1\,10^{-6}\,\mbox{m}. 
\end{array}
\right.
$$
The amplitude $v_0$ of the incident elastic velocity ranges from $10^{-4}$ m/s to $5\,10^{-3}$ m/s. This latter maximal amplitude corresponds to a maximal strain $\varepsilon=v_0/c_0\approx10^{-6}$, so that the linear elastodynamic equations (\ref{LCscal}) are always valid \cite{Achenbach73}. The linear first-order hyperbolic system (\ref{LCscal}) and the jump conditions (\ref{JCsigma}) and (\ref{Model1}) are solved numerically on a $(x,\,t)$ grid. For this purpose, a fourth-order finite-difference ADER scheme is combined with an immersed interface method to account for the jump conditions \cite{Lombard07}. At each time step, numerical integration of $v$ also gives $u$.

Figure \ref{FigNumU} shows snapshots of $u$ after the transients have disappeared and the periodic regime has been reached. The mean values of the incident and reflected displacements ($x<\alpha$) and the transmitted displacement ($x>\alpha$) are given by horizontal dotted lines. With $v_0=10^{-4}$ m/s (a), these mean values are continuous across $\alpha$. At higher amplitudes, a positive jump from $\alpha^-$ to $\alpha^+$ is observed. This jump, which amounts to a mean dilatation of the crack, also increases with $v_0$ (b,c,d). 

These properties can be more clearly seen in figure \ref{FigNumSaut}, where the numerically measured time history of $[u]$ is shown, with $v_0=2\,10^{-3}$ m/s (a) and $v_0=5\,10^{-3}$ m/s (b). It can also be seen from this figure that the maximum value of $[u]$ increases with $v_0$, whereas the minimum value of $[u]$ decreases and is bounded by $-d$, as required by (\ref{Penetration}). These findings will be analysed whatever ${\cal F}$ in (\ref{JCnonlin}) in the following sections.

Distortion of the scattered fields can also be observed in figures \ref{FigNumU} and \ref{FigNumSaut}, increasing with the amplitude of the forcing levels. It is beyond the scope of this paper to sudy this classical nonlinear phenomenon. A local analysis in the case of model 1 (\ref{Model1}) was presented in \cite{Lombard08}.

\begin{figure}[htbp]
\begin{center}
\begin{tabular}{cc}
(a) & (b)\\
\includegraphics[scale=0.31]{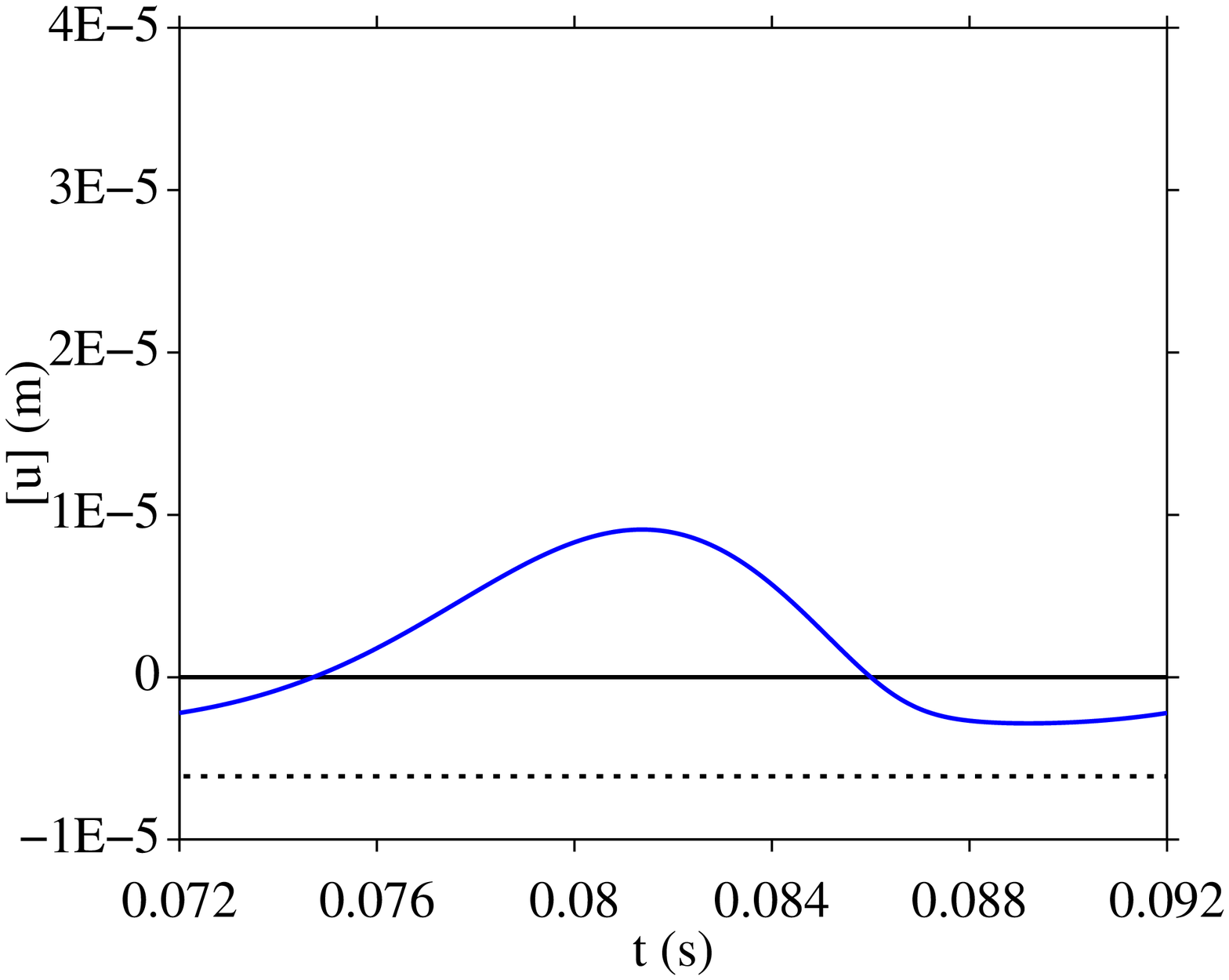}&
\includegraphics[scale=0.31]{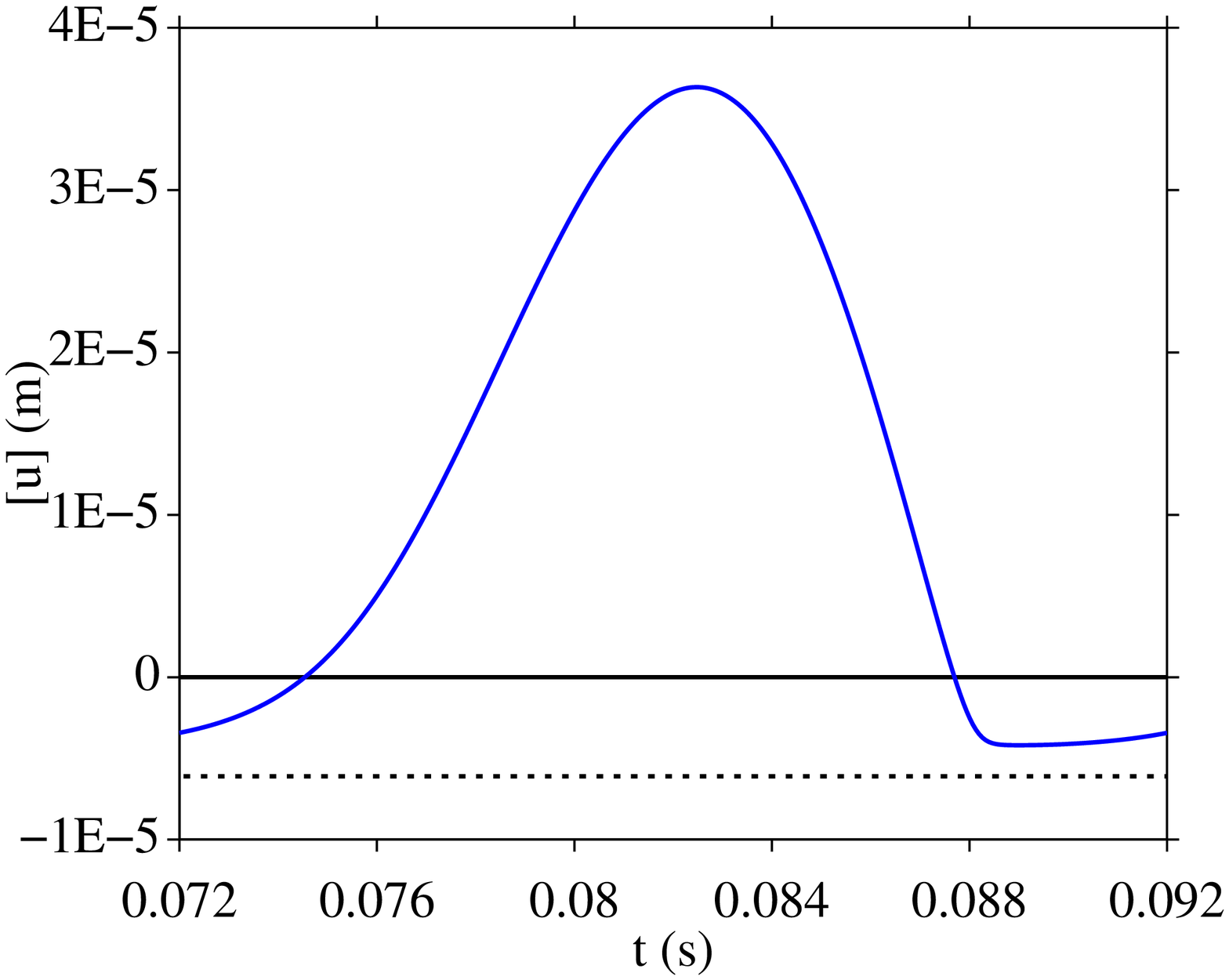}
\end{tabular}
\end{center}
\caption{Time history of $[u(\alpha,\,t)]$, at two incident elastic velocities: $v_0=2\,10^{-3}$ m/s (a) and $v_0=5\,10^{-3}$ m/s (b). The horizontal dotted line denotes $-d$.} 
\label{FigNumSaut}
\end{figure}

%-------------------------------------------------------------------------------------
%-------------------------------------------------------------------------------------

\section{Model problem}\label{SecModelPb}

\subsection{Link with an ODE}

To explain the numerical findings, we look for an evolution equation satisfied by $[u(\alpha,\,t)]$. Considering the elastic wave emitted by the source $S$ in (\ref{LCscal}), which impacts a single crack modeled by (\ref{JCsigma}) and (\ref{JCnonlin}), leads to the following proposition. 

\begin{proposition}
The nondimensionalized jump of displacement $y=[u(\alpha,\,t)]/d$ satisfies the nonautonomous scalar differential equation:
\begin{equation}
\frac{\textstyle d\,y}{\textstyle d\,t}=-{\cal F}(y)+\frac{\textstyle 1}{\textstyle \beta\,d}\,\frac{\textstyle 1}{\textstyle \rho_0\,c_0^2}\,S\left(t/\beta\right), \qquad y(0)=y_0=0,
\label{ODEorig}
\end{equation}
where $\beta=K((\rho_0\,c_0)^{-1}+(\rho_1\,c_1)^{-1})>0$. The time $t$ in (\ref{ODEorig}) has been successively shifted: $t-(\alpha-x_s)/c_0 \rightarrow t$, and rescaled: $t\rightarrow \beta\,t$.
\label{PropositionODE}
\end{proposition}

\begin{proof} 
We adapt a procedure described in \cite{Richardson79,Biwa04,Kim06}. The elastic displacement interacting with the crack is
\begin{equation}
u(x,\,t)=\left\{
\begin{array}{l}
\displaystyle
u_I(t-x/c_0)+u_R(t+x/c_0)\quad \mbox{ if } x<\alpha,\\
[6pt]
u_T(t-x/c_1)\qquad \hspace{1.97cm} \mbox{ if } x>\alpha,
\end{array}
\right.
\label{ODEorig-proofU}
\end{equation}
where $u_I$, $u_R$ and $u_T$ are the incident, reflected and transmitted displacements, respectively. The elastic stress deduced from (\ref{ODEorig-proofU}) is \cite{Achenbach73}
\begin{equation}
\sigma(x,\,t)=\left\{
\begin{array}{l}
\displaystyle
-\rho_0\,c_0\,u_I^{'}(t-x/c_0)+\rho_0\,c_0\,u_R^{'}(t+x/c_0)\quad \mbox{ if } x<\alpha,\\
[6pt]
-\rho_1\,c_1\,u_T^{'}(t-x/c_1)\qquad \hspace{2.75cm} \mbox{ if } x>\alpha.
\end{array}
\right.
\label{ODEorig-proofSig}
\end{equation}
Calculations based on (\ref{ODEorig-proofU}) and (\ref{ODEorig-proofSig}) yield
\begin{equation}
\begin{array}{l}
\displaystyle
\sigma(\alpha^-,\,t)=-2\,\rho_0\,c_0\,u_I^{'}(t-\alpha/c_0)+\rho_0\,c_0\,u^{'}(\alpha^-,\,t),\\
[6pt]
\displaystyle
\sigma(\alpha^+,\,t)=-\rho_1\,c_1\,u^{'}(\alpha^+,\,t).
\end{array}
\label{ODEorig-proofSigpm}
\end{equation}
At $\alpha$, the jump in the displacement $\chi$ and the weighted displacement $\psi$ are introduced
$$
\begin{array}{l}
\displaystyle
\chi(t)=[u(\alpha,\,t)]=u(\alpha^+,\,t)-u(\alpha^-,\,t),\\
[6pt]
\displaystyle
\psi(t)=\frac{\textstyle 1}{\textstyle 2}\left(u(\alpha^-,\,t)+\frac{\textstyle \rho_1\,c_1}{\textstyle \rho_0\,c_0}\,u(\alpha^+,\,t)\right),
\end{array}
$$
or, alternatively,
\begin{equation}
\begin{array}{l}
\displaystyle
u(\alpha^+,\,t)=\frac{\textstyle 2\,\rho_0\,c_0}{\textstyle \rho_0\,c_0+\rho_1\,c_1}\,\psi(t)+\frac{\textstyle \rho_0\,c_0}{\textstyle \rho_0\,c_0+\rho_1\,c_1}\,\chi(t),\\
[10pt]
\displaystyle
u(\alpha^-,\,t)=\frac{\textstyle 2\,\rho_0\,c_0}{\textstyle \rho_0\,c_0+\rho_1\,c_1}\,\psi(t)-\frac{\textstyle \rho_1\,c_1}{\textstyle \rho_0\,c_0+\rho_1\,c_1}\,\chi(t).
\end{array}
\label{ODEorig-proofPhi}
\end{equation}
The continuity of $\sigma$ (\ref{JCsigma}) means
\begin{equation}
\frac{\textstyle d\,\psi}{\textstyle d\,t}=u_I^{'}(t-\alpha/c_0).
\label{ODEorig-proofdPsi}
\end{equation}
The nonlinear relation (\ref{JCnonlin}) and the first equation of (\ref{ODEorig-proofSigpm}) give
$$
u^{'}(\alpha^-,\,t)=2\,u_I^{'}(t-\alpha/c_0)+\frac{\textstyle K\,d}{\textstyle \rho_0\,c_0}\,{\cal F}(\chi/d).
$$
Using (\ref{ODEorig-proofPhi}) and (\ref{ODEorig-proofdPsi}), we obtain after some operations
$$
\frac{\textstyle d\,\chi}{\textstyle d\,t}=-\beta\,d\,{\cal F}(\chi/d)-2\,u^{'}_I(t-\alpha/c_0), 
$$
where $\beta$ is defined in proposition \ref{PropositionODE}. Classical calculations of elastodynamics yield the incident elastic velocity generated by $S$ and impacting the crack \cite{Achenbach73}
\begin{equation}
v(x,\,t)=u_I^{'}(t-x/c_0)=-\frac{\textstyle 1}{2\,\rho_0\,c_0^2}S\left(t-\frac{\textstyle x-x_s}{\textstyle c_0}\right).
\label{Vinc}
\end{equation}
The time shift $t-(\alpha-x_s)/c_0 \rightarrow t$ therefore gives the differential equation
$$
\frac{\textstyle d\,\chi}{\textstyle d\,t}=-\beta\,d\,{\cal F}(\chi/d)+\frac{\textstyle 1}{\textstyle \rho_0\,c_0^2}\,S(t).
$$
Nondimensionalization $y=\chi/d$ and time scaling $t\rightarrow \beta\,t$ give the ODE (\ref{ODEorig}). Since the source $S$ is causal and $x_s<\alpha$, then $y(0)=0$.
\qquad\end{proof}

%-------------------------------------------------------------------------------------

\subsection{Sinusoidal forcing}\label{SecSinus}

From $t=0$, the source is assumed to be sinusoidal with the angular frequency $\Omega$
\begin{equation}
S(t)=2\,v_0\,\rho_0\,c_0^2\,\sin \Omega t,
\label{ForcageSinus}
\end{equation}
which results in an incident elastic velocity with the amplitude $v_0$ (\ref{Vinc}). The source (\ref{ForcageSinus}) is injected in (\ref{ODEorig}). Then, setting the nondimensionalized parameters 
\begin{equation}
\begin{array}{l}
\displaystyle
A=\frac{\textstyle 2\,v_0}{\textstyle \beta\,d},\qquad\omega=\frac{\textstyle \Omega}{\textstyle \beta},\qquad T=\frac{\textstyle 2\,\pi}{\textstyle \omega}=\frac{\textstyle 2\,\pi\,\beta}{\textstyle \Omega},\\%\qquad 
[8pt]
\displaystyle
y=\frac{\textstyle [u(\alpha,\,t)]}{\textstyle d},\quad f(y)=-{\cal F}(y),\qquad f_{\min}=-{\cal F}_{\max},
\end{array}
\label{Adimentionalize}
\end{equation}
the model problem becomes a nonautonomous differential equation with sinusoidal forcing:
\begin{equation}
\left\{
\begin{array}{l}
\displaystyle
\frac{\textstyle d\,y}{\textstyle d\,t}=f(y)+A\,\sin\omega t=F(t,\,y),\qquad \omega>0,\qquad 0\leq A< +\infty,\\
[6pt]
\displaystyle
f:\,]-1,\,+\infty[\rightarrow]f_{\min},\,+\infty[,\quad\lim_{y\rightarrow-1}f(y)=+\infty,\quad -\infty\leq f_{\min}<0,\\
[6pt]
\displaystyle
f(0)=0,\quad f^{'}(0)=-1,\quad f^{'}(y)<0<f^{''}(y), \\
[6pt]
\displaystyle
y(0)=y_0\in]-1,\,+\infty[.
\end{array}
\right.
\label{ODE_Y}
\end{equation}
For the sake of generality, contrary to what occured in proposition \ref{PropositionODE}, $y_0$ can differ here from 0. The properties of $f$ in (\ref{ODE_Y}) mean that the reciprocal function $f^{-1}$ satisfies
\begin{equation}
\begin{array}{l}
\displaystyle
f^{-1}:\,]f_{\min},\,+\infty[\rightarrow]-1,\,+\infty[,\qquad\lim_{y\rightarrow+\infty}f^{-1}(y)=-1,\\
[6pt]
\displaystyle
f^{-1}(0)=0,\quad (f^{-1})^{'}<0<(f^{-1})^{''}.
\end{array}
\label{PropFm1}
\end{equation}
In model 1 (\ref{Model1}), $f$ is involutive and $f_{\min}=-1$. The function $f$ and its derivatives are
$$
f(y)=-\frac{\textstyle y}{\textstyle 1+y},\quad f^{'}(y)=-\frac{\textstyle 1}{\textstyle (1+y)^2},\quad f^{''}(y)=\frac{\textstyle 2}{\textstyle (1+y)^3}, \quad f^{'''}(y)=-\frac{\textstyle 6}{\textstyle (1+y)^4}.
$$
In model 2 (\ref{Model2}), $f_{\min}=-\infty$, and the function $f$ and its derivatives are
$$
f(y)=-\ln(1+y),\quad f^{'}(y)=-\frac{\textstyle 1}{\textstyle 1+y},\quad f^{''}(y)=\frac{\textstyle 1}{\textstyle (1+y)^2},\quad f^{'''}(y)=-\frac{\textstyle 2}{\textstyle (1+y)^3}.
$$
In the next sections, solutions of (\ref{ODE_Y}) and related ODE will be presented. These solutions are computed numerically using a fourth-order Runge-Kutta method, with model 1 (\ref{Model1}). The parameters are those used in section \ref{SecNumExp}, with nondimensionalization (\ref{Adimentionalize}). In addition and for the sake of clarity, the dependence of solutions on the parameters $A$ and $\omega$ is omitted, except when necessary. Lastly, we use an overline to denote the mean value of a function during  one period $[0,\,T]$: given $s(t)$,
$$
\overline{s}=\frac{\textstyle 1}{\textstyle T}\int_0^Ts(t)\,dt.
$$

%-------------------------------------------------------------------------------------
%-------------------------------------------------------------------------------------

\section{Preliminary results}\label{SecPreRes}

\subsection{Existence and uniqueness of a periodic solution $Y$}

The $T$-periodic isocline of zero slope deduced from (\ref{ODE_Y}) is
\begin{equation}
I_0^y(t)=f^{-1}(-A\,\sin \omega t).
\label{IsoclineY}
\end{equation}

\begin{proposition}
There is a unique $T$-periodic solution $Y(t)$ of (\ref{ODE_Y}). This solution is asymptotically stable.
\label{PropositionYperiod}
\end{proposition}

\begin{proof}
For the sequel, we denote $Y_0=Y(0)$ the unique initial data such that the solution of (\ref{ODE_Y}) is $T$-periodic. Three cases can be distinguished.

\underline{Case 1}: $A=0$. Equation (\ref{ODE_Y}) becomes a scalar autonomous equation with the steady state solution $Y=0$. Since $f$ is a $C^1$ function and $f^{'}(0)=-1$, the fixed point 0 is asymptotically stable. In the case of model 1 (\ref{Model1}), the exact solution is a known special function: $y\,e^y=y_0\,e^{y_0-t}\,\Rightarrow \,y(t)={\cal L}\left(y_0\,e^{y_0-t}\right)$, where ${\cal L}$ is the Lambert function \cite{Lambert96}.

\underline{Case 2}: $0<A<|f_{\min}|$. The $T$-periodic isocline (\ref{IsoclineY}) is defined on $\mathbb{R}$, with lower and upper bounds $f^{-1}(A)<0$ and $f^{-1}(-A)>0$ (figure \ref{FigIsoclineY}-a). Below $I_0^y$, we obtain $\frac{d\,y}{d\,t}>0$; and above $I_0^y$, we obtain $\frac{d\,y}{d\,t}<0$. Consequently, the null slope of the horizontal line $f^{-1}(A)$ is smaller than $F(t,\,y)$ when these curves intersect, except at $t=3\,T/4$, where they are both equal to zero: $f^{-1}(A)$ is called a {\it weak lower fence} for (\ref{ODE_Y}) \cite{Hubbard91}. Similarly, $F(t,\,y)\leq(f^{-1}(-A))^{'}=0$, where equality occurs only at $T/4$: $f^{-1}(-A)$ is called a {\it weak upper fence} for (\ref{ODE_Y}). Since $y\mapsto F(t,\,y)$ is Lipschitz continuous, $f^{-1}(\pm A)$ are {\it nonporous fences}: once a solution has crossed these fences, it cannot sneak through them \cite{Hubbard91}. Based on the {\it funnel's theorem}, the subset of the $(t,\,y)$ plane defined by
$$
K=\mathbb{R}^+ \times K_0,\quad K_0=\left[f^{-1}(A),\,f^{-1}(-A)\right],  
$$
is a funnel \cite{Hubbard91}: once a solution has entered $K$, it stays inside. Defining the subsets 
$$
U^+=\left\{t\geq 0;\,-1<y<f^{-1}(A)\right\},\qquad U^-=\left\{t\geq 0;\,y>f^{-1}(-A)\right\},
$$
the solution of (\ref{ODE_Y}) increases in $U^+$ and decreases in $U^-$ as $t$ increases. In addition, zero slope $\frac{d\,y}{d\,t}=0$ in $U^\pm$ occurs only at $T/4$ and $3\,T/4$. The  solution therefore enters $K$ in finite time, and $K_0$ is invariant under the flow $\varphi:\mathbb{R}^+\times]-1,\,+\infty[\rightarrow ]-1,\,+\infty[;\,(t,\,y_0)\mapsto y(t)=\varphi(t,\,y_0)$. The solutions of (\ref{ODE_Y}) are therefore always bounded, which ensures that the solutions are global solutions in time.

The invariance of $K_0$ ensures $\Pi(K_0)\subset K_0$, with the Poincar\'e map $\Pi:]-1,\,+\infty[\rightarrow]-1,\,+\infty[;\,y_0\mapsto \Pi(y_0)=\varphi(T,\,y_0)$. Since $\Pi$ is continuous on $K$, it has at least one fixed point $Y_0$. The solutions $Y(t)=\varphi(t,\,Y_0)$ and $Y_T(t)=\varphi(t+T,\,Y_0)$ both satisfy the same ODE with the same initial value: $Y_T(0)=\varphi(T,\,Y_0)=\Pi(Y_0)=Y_0=Y(0)$. Therefore $Y(t+T)=Y(t)$ at all values of $t$, which proves that $Y(t)$ is $T$-periodic.

\begin{figure}[htbp]
\begin{center}
\begin{tabular}{cc}
(a) & (b)\\
\includegraphics[scale=0.31]{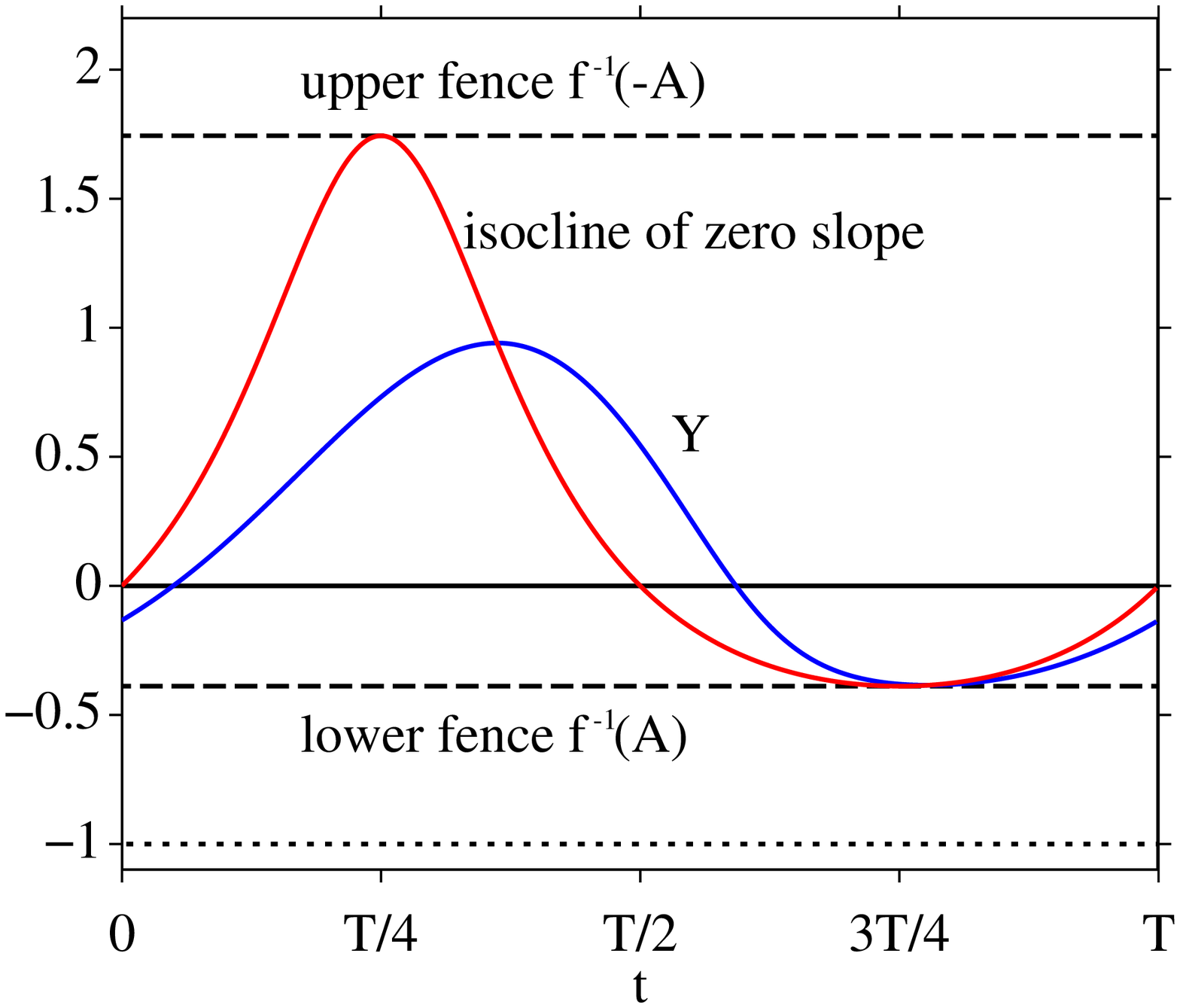}&
\includegraphics[scale=0.31]{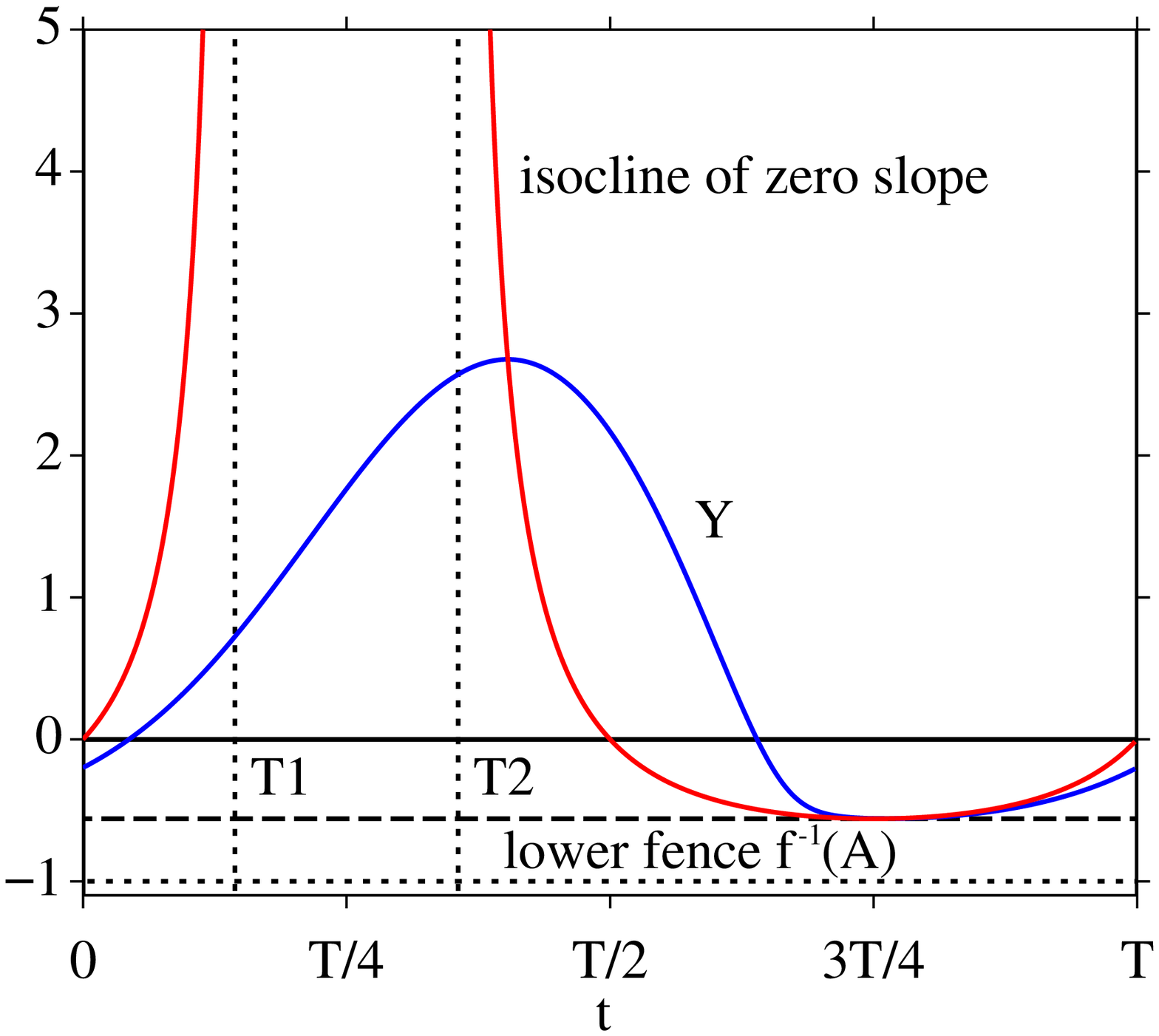}
\end{tabular}
\end{center}
\caption{Phase portrait of (\ref{ODE_Y}) with model 1 (\ref{Model1}), at $A=0.6$ (a) and $A=1.2$ (b): isocline of zero slope $I_0^y$ and periodic solution $Y$. Dashed horizontal lines: lower fence $f^{-1}(A)<0$, upper fence $f^{-1}(-A)>0$ (a). The dotted horizontal line stands for $y=-1$. In (b), the dotted vertical lines stand for asymptotes at $T_1$ and $T_2$.} 
\label{FigIsoclineY}
\end{figure}

Let us consider another $T$-periodic solution ${\tilde Y}(t)$ of (\ref{ODE_Y}), with ${\tilde Y}(0)={\tilde Y}_0$. Assuming that $Y_0>{\tilde Y}_0$ yields $Y(t)>{\tilde Y}(t)$ at all values of $t\in[0,\,T]$, hence $f(Y(t))<f({\tilde Y}(t))$. The integral forms of $Y(T)$ and ${\tilde Y}(T)$, together with the $T$-periodicity, give
\begin{equation}
\begin{array}{lll}
\displaystyle
Y_T-{\tilde Y}_T&=&\displaystyle Y_0-{\tilde Y}_0+\int_0^T\left(f(Y(t))-f({\tilde Y}(t)\right)\,dt<Y_0-{\tilde Y}_0,\\
[8pt]
&=& \displaystyle Y_0-{\tilde Y}_0.
\end{array}
\label{ProofYperiodUnique}
\end{equation}
This results in $Y_0-{\tilde Y}_0<Y_0-{\tilde Y}_0$, which is impossible. The opposite case $Y_0<{\tilde Y}_0$ can be handled in the same way, which proves the uniqueness of the $T$-periodic solution. Lastly, the properties of $f$ in (\ref{ODE_Y}) mean that
\begin{equation}
a_0=\int_0^T\frac{\textstyle \partial\,F}{\textstyle \partial\,y}(t,\,\varphi(t,\,y_0))\,dt=\int_0^Tf^{'}(\varphi(t,\,y_0))\,dt<0,
\label{ProofYperiodA0}
\end{equation}
therefore $Y(t)$ is asymptotically stable (see e.g. theorem 4-22 in \cite{HaleKocak91}).

\underline{Case 3}: $A \geq |f_{\min}|$. This case can occur only if $f_{\min}>-\infty$. The isocline (\ref{IsoclineY}) is defined outside the two vertical asymptotes at $T_1$ and $T_2$
\begin{equation}
0<T_1=\frac{\textstyle 1}{\textstyle \omega}\,\arcsin\left(\frac{\textstyle |f_{\min}|}{\textstyle A}\right)<\frac{\textstyle T}{\textstyle 4}<T_2=\frac{\textstyle T}{\textstyle 2}-T_1.
\label{AsymptotsY}
\end{equation}
(figure \ref{FigIsoclineY}-b). The set $[f^{-1}(A),\,+\infty[$ is invariant under the flow but not compact. To prove the existence of solutions at all time, let us examine the subset of the $(t,\,y)$ plane where the solution can blow up: $\Gamma=\left\{0<t<T/2;\,y>0\right\}$. In $\Gamma$, $y>0$ and $f(y)<0$, hence $\frac{d\,y}{d\,t}<A\,\sin \omega t$. By integrating, we obtain
$$
0<y<y_0+\int_0^{T/2}A\,\sin \omega t\,dt=y_0+\frac{\textstyle 2\,A}{\textstyle \omega}<+\infty, 
$$
which proves that the solution is global in time.

To obtain an invariant set, we take $y>0$, which means that $f_{\min}<f(y)<0$ and thus $f_{\min}+A\,\sin \omega t<\frac{d\,y}{d\,t}<A\,\sin \omega t$. As long as $y>0$, integration will  ensure that $y$ is bounded by a lower solution and an upper solution
$$
y_0+f_{\min}\,t+\frac{\textstyle A}{\textstyle \omega}\,\left(1-\cos \omega t\right)<y(t)<y_0+\frac{\textstyle A}{\textstyle \omega}\,\left(1-\cos \omega t\right).
$$ 
Taking $y_0>|f_{\min}|\,T$, for instance $y_0=2\,|f_{\min}|\,T$, ensures that $y>0$ for all values of $t\in[0,\,T]$ and that $y(T)<y_0$. Therefore, upon defining the subset of the $(t,\,y)$ plane
$$
K=\mathbb{R}^+ \times K_0,\quad K_0=\left[f^{-1}(A),\,2\,|f_{\min}|\,T\right],   
$$
then $K_0$ is invariant under the Poincar\'e map, yielding $\Pi(K_0)\subset K_0$. The same arguments as those given in case 2 hold here, which proves the existence, uniqueness and asymptotic stability of the $T$-periodic solution. 
\qquad\end{proof}\\

As seen in the previous proof, the phase portrait of (\ref{ODE_Y}) depends on $A$. However, the evolution of $Y$ with $t$ is the same at all values of $A$, as seen in figure \ref{FigIsoclineY} and proved in section \ref{SecQualitY}. Some auxiliary solutions are first introduced. These solutions make it possible to investigate how the attractive periodic solution $Y$ evolves with the parameter $A$.

%-------------------------------------------------------------------------------------

\subsection{Auxiliary solutions}\label{SecAuxiliary}

The first derivative of $Y$ with respect to $A$ is introduced: $Z(t,\,A)=\frac{\partial\,Y}{\partial\,A}$. Applying the chain-rule to (\ref{ODE_Y}) shows that the $T$-periodic solution $Z$ satisfies 
\begin{equation}
\left\{
\begin{array}{l}
\displaystyle
\frac{\textstyle d\,Z}{\textstyle d\,t}=f^{'}(Y)\,Z+\sin \omega t,\\
\displaystyle
Z(0,\,A)=Z_0=\frac{\textstyle d\,Y_0}{\textstyle d\,A}.
\displaystyle
\end{array}
\right.
\label{ODE_Z}
\end{equation}
The isocline of zero slope of (\ref{ODE_Z}) is deduced from (\ref{ODE_Y}):
\begin{equation}
I_0^Z(t)=-\frac{\textstyle \sin \omega t}{\textstyle f^{'}(Y(t))}=\frac{\textstyle \sin \omega t}{\textstyle \left|f^{'}(Y(t))\right|}.
\label{IsoclineZ}
\end{equation}
Below $I_0^Z$, we obtain $\frac{d\,Z}{d\,t}>0$; and above $I_0^Z$, we obtain $\frac{d\,Z}{d\,t}<0$. The phase portrait of (\ref{ODE_Z}) obtained in the case of model 1 is displayed in figure \ref{FigIsoclineZ}, which shows the qualitative properties of $Z$ that will be described in the next lemma and proposition.

\begin{figure}[htbp]
\begin{center}
\begin{tabular}{cc}
(a) & (b)\\
\includegraphics[scale=0.31]{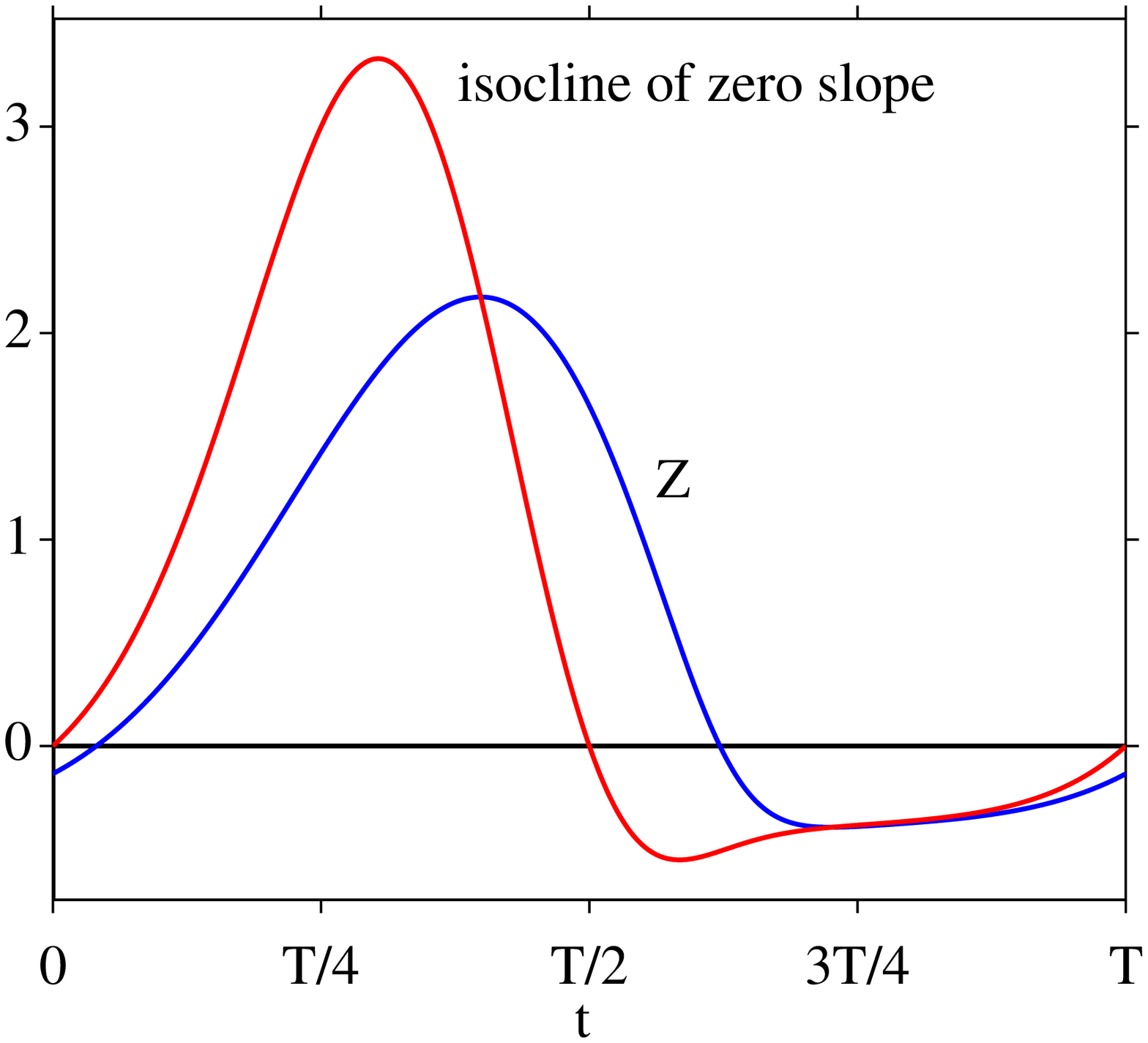}&
\includegraphics[scale=0.31]{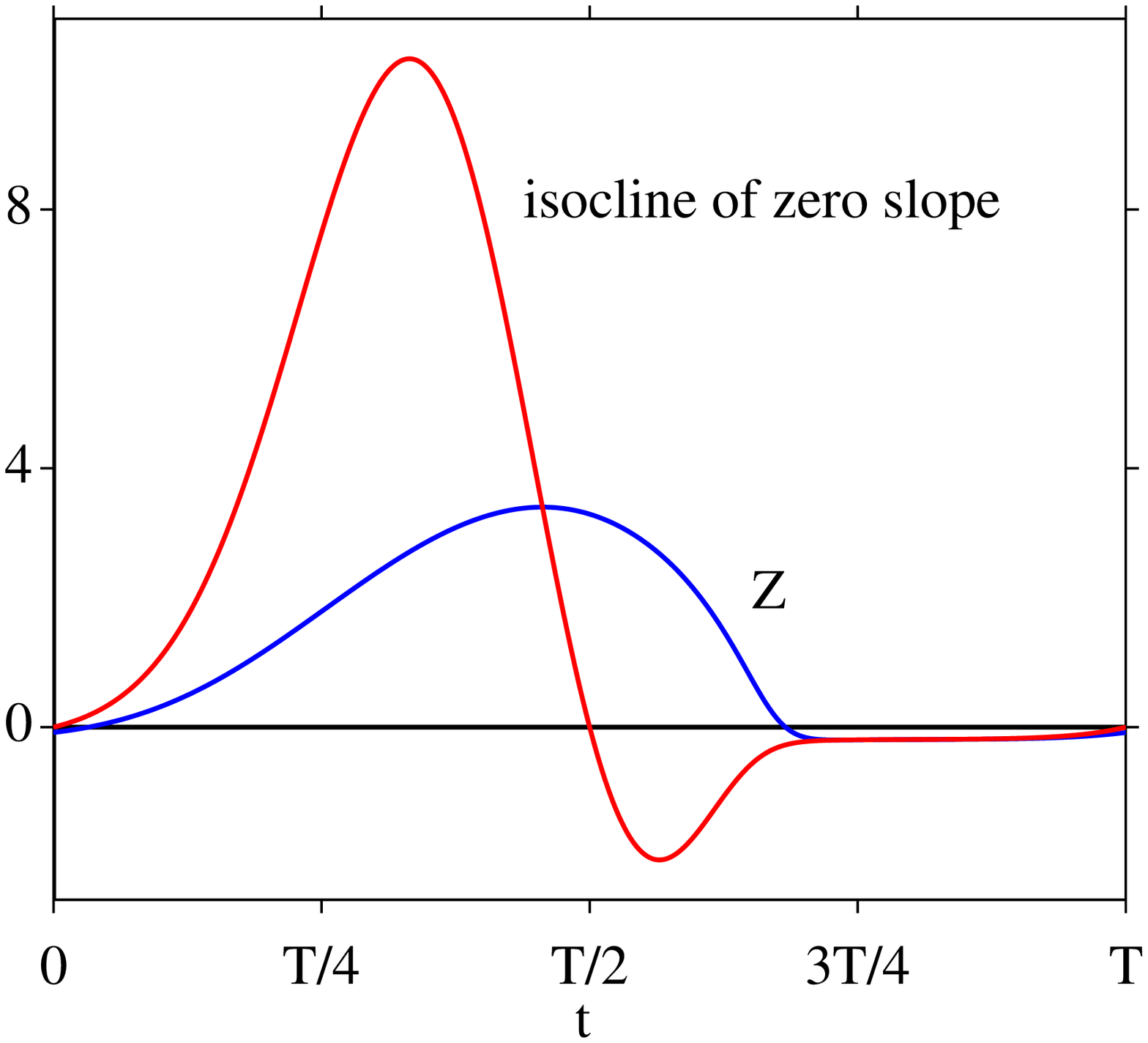}
\end{tabular}
\end{center}
\caption{Phase portrait of (\ref{ODE_Z}) with model 1 (\ref{Model1}), at $A=0.6$ (a) and $A=1.2$ (b): isocline of zero slope $I_0^Z$ and periodic solution $Z$.} 
\label{FigIsoclineZ}
\end{figure}

\begin{lemma}
Setting
$$
B(s)=-\int_0^s f^{'}(Y(\tau))\,d\tau>0,
$$
the initial value of the $T$-periodic solution $Z$ in (\ref{ODE_Z}) satisfies 
\begin{equation}
\displaystyle
Z_0=\frac{\textstyle 1}{\textstyle e^{B(T)}-1}\int_0^Te^{B(s)}\,\sin\omega s\,ds<0.
\label{Z_0}
\end{equation}
\label{LemmaZ0}
\end{lemma}

\begin{proof}
Given $Y(t)$, the ODE (\ref{ODE_Z}) is linear, which leads to the closed-form expression
$$
Z(t)=\left(D+\int_0^Te^{B(s)}\,\sin \omega s\,ds\right)\,e^{-B(t)}.
$$
$T$-periodicity of $Z$ gives $D$ and proves the equality established in (\ref{Z_0}). Since $s\mapsto e^{B(s)}$ is a strictly positive and increasing function of $s$, we obtain the following bound
$$
\int_0^Te^{B(s)}\,\sin \omega s\,ds< \displaystyle e^{B(T/2)}\int_0^{T/2}\sin \omega s\,ds+e^{B(T/2)}\int_{T/2}^T\sin \omega s\,ds=0.
$$
Substituting this inequality into (\ref{Z_0}) proves that $Z_0<0$.
\qquad\end{proof}\\

\begin{proposition}
In $[0,\,T]$, the $T$-periodic solution $Z$ of (\ref{ODE_Z}) has two roots $t_{Z_1}$ and $t_{Z_2}$, which are ordered as follows:
$$
0<t_{Z_1}<T/2<t_{Z_2}<T.
$$
Therefore $Z<0$ in $[0,\,t_{Z_1}[\cup ]t_{Z_2}\,T]$, and $Z>0$ on $]t_{Z_1},\,t_{Z_2}[$. In the limit of null-forcing case $A=0$, $Z$ and its roots are determined analytically:
\begin{equation}
\left\|
\begin{array}{l}
\displaystyle
Z(t,\,0)=\frac{\textstyle 1}{\textstyle \sqrt{1+\omega^2}}\,\sin(\omega t-\theta),\qquad\theta=\arctan \omega,\\
\\
\displaystyle
t_{Z_1}(0)=\frac{\textstyle \theta}{\textstyle \omega},\qquad t_{Z_2}(0)=\frac{\textstyle \theta}{\textstyle \omega}+\frac{\textstyle T}{\textstyle 2}.
\label{Zeps0}
\end{array}
\right.
\end{equation}
\label{PropositionTz}
\end{proposition}

\begin{proof}
$T$-periodicity of $Z$ and integration of (\ref{ODE_Z}) over one period gives
$$
\int_0^T f^{'}(Y(\tau))\,Z(\tau)\,d\tau=0.
$$
Since $f^{'}<0$, the sign of $Z$ changes in $[0,\,T]$. From lemma \ref{LemmaZ0} and the phase portrait of (\ref{ODE_Z}), it can be deduced that $Z$ must cross the $t$-axis twice, at $t_{Z_1}$ and $t_{Z_2}$ in $]0,\,T[$. Let us assume $t_{Z_1}\geq T/2$, where the isocline (\ref{IsoclineZ}) is negative (figure \ref{FigIsoclineZ}). 
Then $Z$ will intersect $I_0^Z$ twice at negative values and will never cross the $t$-axis, which is impossible. After $t_{Z_1}$, $Z$ increases and intersects $I_0^Z$ with a zero slope at $t_{Z_1}<t<T/2$. It then decreases but cannot intersect again $I_0^Z$ as long as $I_0^Z>0$ i.e. at $t<T/2$. This means that $t_{Z_2}>T/2$.

Taking $A=0$ results in $Y=0$, as stated in the proof of proposition \ref{PropositionYperiod} (case 1), and hence (\ref{ODE_Z}) becomes the linear ODE
$$
\displaystyle Z^{'}=-Z+\sin \omega t,
$$
the $T$-periodic solution of which is given in (\ref{Zeps0}). 
\qquad\end{proof}\\

\begin{figure}[htbp]
\begin{center}
\begin{tabular}{cc}
(a) & (b)\\
\includegraphics[scale=0.31]{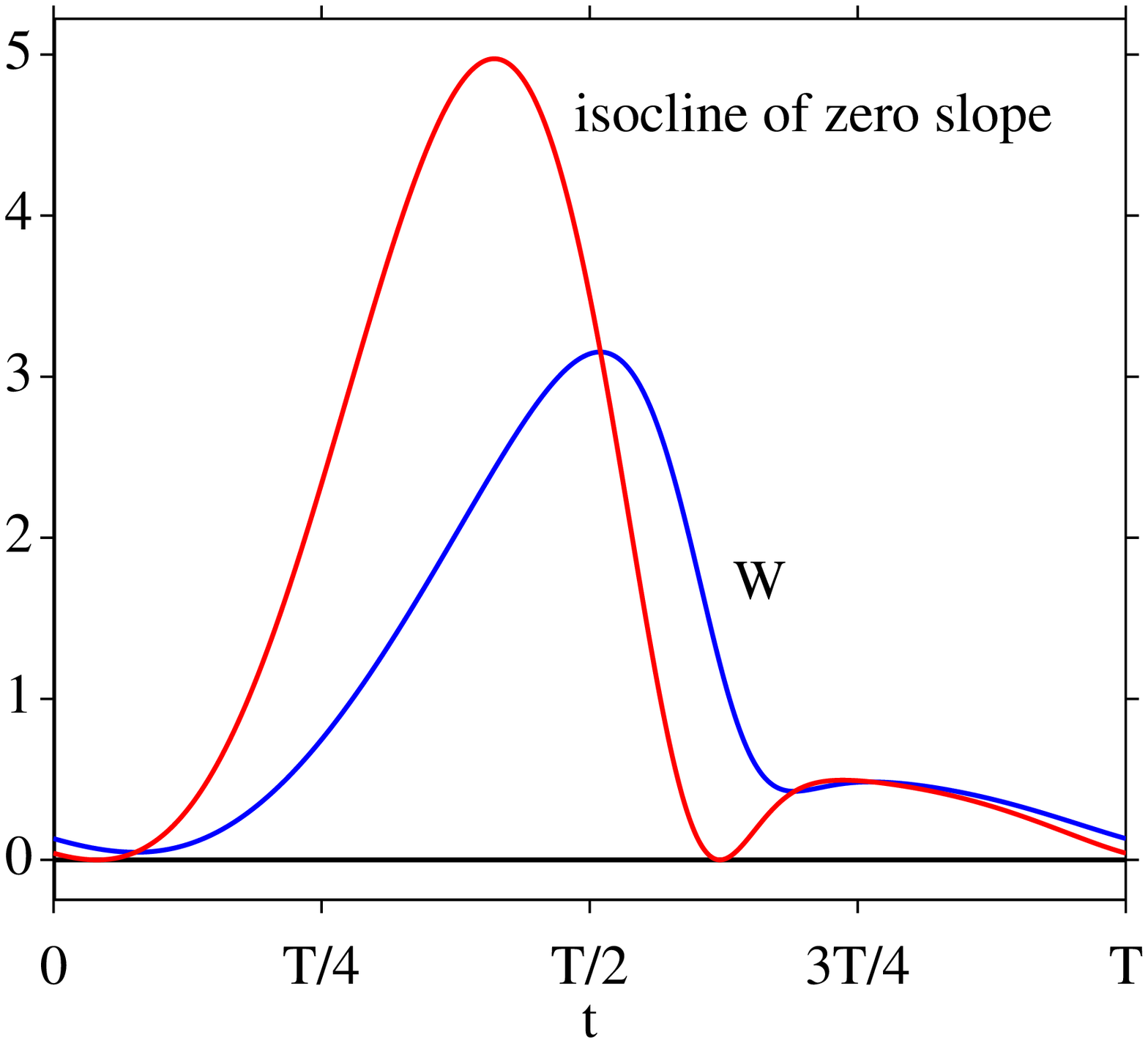}&
\includegraphics[scale=0.31]{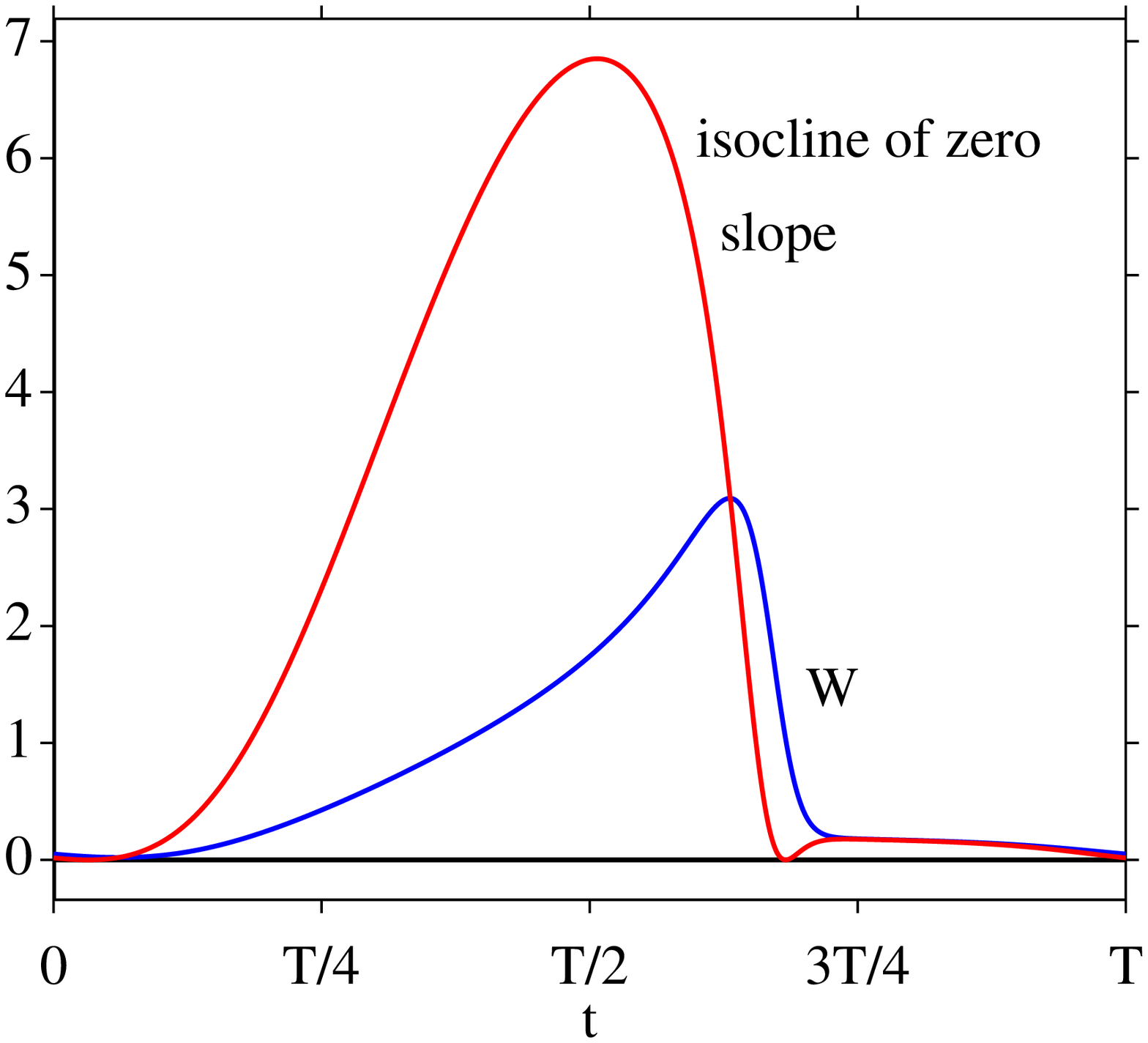}
\end{tabular}
\end{center}
\caption{Phase portrait of (\ref{ODE_W}) with model 1 (\ref{Model1}), at $A=0.6$ (a) and $A=1.2$ (b): isocline of zero slope $I_0^W$ and periodic solution $W$.} 
\label{FigIsoclineW}
\end{figure}

The second derivative of $Y$ with respect to $A$ is also introduced: $W(t,\,A)=\frac{\partial^2\,Y}{\partial\,A^2}$. Applying the chain-rule to (\ref{ODE_Z}) shows that the $T$-periodic solution $W$ satisfies
\begin{equation}
\left\{
\begin{array}{l}
\displaystyle
\frac{\textstyle d\,W}{\textstyle d\,t}=f^{''}(Y)\,Z^2+f^{'}(Y)\,W,\\
\displaystyle
W(0,\,A)=W_0=\frac{\textstyle d^2\,Y_0}{\textstyle d\,A^2}.
\displaystyle
\end{array}
\right.
\label{ODE_W}
\end{equation}
The isocline of zero slope of $W$, which is deduced from (\ref{ODE_Y}) and (\ref{ODE_W}), is 
\begin{equation}
I_0^W(t)=\frac{\textstyle f^{''}(Y(t))}{\textstyle \left|f^{'}(Y(t))\right|}\,Z^2(t)\geq 0.
\label{IsoclineW}
\end{equation}
This isocline, which is shown in figure \ref{FigIsoclineW}, is equal to zero only at $t_{Z_1}$ and $t_{Z_2}$, when $Z$ vanishes (proposition \ref{PropositionTz}). Below $I_0^W$, we obtain $\frac{d\,W}{d\,t}>0$; and above $I_0^W$, we obtain $\frac{d\,W}{d\,t}<0$.

\begin{proposition}
The $T$-periodic solution $W$ of (\ref{ODE_W}) is strictly positive: for all $t\in[0,\,T]$, $W(t,\,A)>0$.
\label{PropositionW}
\end{proposition}

\begin{proof}
If $W_0\leq 0$, the phase portrait of (\ref{ODE_W}) entails that $W(T)>W_0$, which is contradictory with the $T$-periodicity of $W$: hence, $W_0>0$. We then consider 
$$
\left\{
\begin{array}{l}
w^{'}=f^{'}(Y)\,w,\\
[6pt]
0<w(0)<W_0.
\end{array}
\right.
$$
Since $f{''}>0$ and $W_0>w(0)$, $w$ is a lower solution of $W$, hence $W(t)>w(t)$ at all $t$. However, the $T$-periodic lower solution is $w=0$, which completes the proof.
\qquad\end{proof}\\

\begin{lemma}
In the limit case of null forcing, the $T$-periodic solution $W$ of (\ref{ODE_W}) contains only the harmonics $\sin 2\omega t$ and $\cos 2\omega t$, and has a non-null mean value 
\begin{equation}
\overline{W}(0)=\frac{\textstyle f^{''}(0)}{\textstyle 2}\,\frac{\textstyle 1}{\textstyle 1+\omega^2}.
\label{Weps0}
\end{equation}
\label{LemmaW0}
\end{lemma}

\begin{proof}
We inject $Y=0$ into (\ref{ODE_Y}) and (\ref{ODE_W}). Using (\ref{Zeps0}) gives
$$
W^{'}(t,\,0)=\frac{\textstyle f^{''}(0)}{\textstyle 1+\omega^2}\,\sin^2(\omega t-\theta)-W(t,\,0).
$$
Integrating the latter equation over $[0,\,T]$ gives the result.
\qquad\end{proof}\\

Lastly, the third derivative of $Y$ with respect to $A$ is introduced: $X(t,\,A)=\frac{\partial^3\,Y}{\partial\,A^3}$. Applying the chain-rule to (\ref{ODE_W}) shows that the $T$-periodic solution $X$ satisfies
\begin{equation}
\left\{
\begin{array}{l}
\displaystyle
\frac{\textstyle d\,X}{\textstyle d\,t}=f^{'''}(Y)\,Z^3+3\,f^{''}(Y)\,Z\,W+f^{'}(Y)\,X,\\
\displaystyle
X(0,\,A)=X_0=\frac{\textstyle \partial^3\,Y_0}{\textstyle \partial\,A^3}.
\displaystyle
\end{array}
\right.
\label{ODE_X}
\end{equation}

\begin{lemma}
In the limit case of null forcing, the mean value of the $T$-periodic solution $X$ of (\ref{ODE_X}) is $\overline{X}(0)=0$.
\label{LemmaX0}
\end{lemma}

\begin{proof}
Upon substituting $Y=0$ into (\ref{ODE_Y}), then (\ref{ODE_X}) becomes
$$
X^{'}(t,\,0)=f^{'''}(0)\,Z^3(t,\,0)+3\,f^{''}(0)\,Z(t,\,0)\,W(t,\,0)-X(t,\,0).
$$
Integrating this differential equation over $[0,\,T]$ and the $T$-periodicity of $X$ give
$$
\overline{X}(0)=f^{'''}(0)\,\overline{Z(t,\,0)^3}+3\,f^{''}(0)\,\overline{Z(t,\,0)\,W(t,\,0)}.
$$
In (\ref{Zeps0}), $Z(t,\,0)$ is a sinusoidal function with a null mean value, hence $\overline{Z(t,0)^3}=0$. Likewise, (\ref{Zeps0}) and the properties of $W(t,\,0)$ stated in lemma \ref{LemmaW0} give $\overline{Z(t,\,0)\,W(t,\,0)}=0$, which proves the result obtained.
\qquad\end{proof}

\vspace{1cm}

%-------------------------------------------------------------------------------------
%-------------------------------------------------------------------------------------

\section{Qualitative properties of $Y$}\label{SecQualitY}

\subsection{Phase portrait of $Y$}

The properties of $Z$ and $W$ provide qualitative insight about $Y$, stated in the next propositions and corollary.

\begin{figure}[htbp]
\begin{center}
\begin{tabular}{c}
%(a)\\
\includegraphics[scale=0.31]{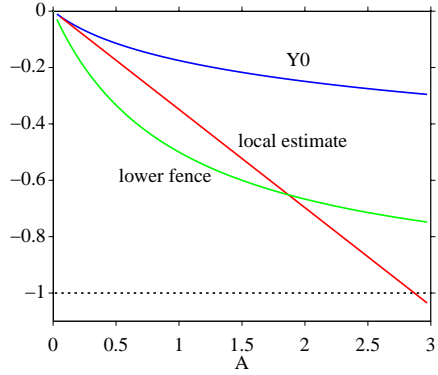}
\end{tabular}
\end{center}
\caption{Evolution of $Y_0$ with the forcing amplitude $A$ in the case of model 1 (\ref{Model1}). See proposition \ref{PropositionY0}.} 
\label{FigY0}
\end{figure}

\begin{proposition}
The initial value $Y_0$ in (\ref{ODE_Y}) is a strictly decreasing convex function of $A$, which satisfies 
$$
-1<\max\left(f^{-1}(A),\,-\frac{\textstyle A\,\omega}{\textstyle 1+\omega^2}\right) \leq Y_0(A) \leq 0.
$$
Equality $Y_0=0$ occurs only at $A=0$. At low forcing levels, we obtain the local estimate
$$
Y_0=-\frac{\textstyle A\,\omega}{\textstyle 1+\omega^2}+{\cal O}(A^2).
$$
\label{PropositionY0}
\end{proposition}

\begin{proof}
Propositions \ref{PropositionTz} and \ref{PropositionW} yield $\frac{d\,Y_0}{d\,A}=Z(0,\,A)<0$ and $\frac{d^2\,Y_0}{d\,A^2}=W(0,\,A)>0$, which proves the convex decreasing of $Y_0$. Since $Y_0(0)=Y(0,\,0)=0$, then $Y_0(A)\leq 0$, which gives the upper bound of $Y_0$.

We now prove the lower bounds of $Y_0$. In the proof of proposition \ref{PropositionYperiod}, it was established that $f^{-1}(A)$ is a lower nonporous fence, which means that $Y(t)\geq f^{-1}(A)$, and hence that $Y_0\geq f^{-1}(A)$. On the other hand, lemma \ref{LemmaZ0} and proposition \ref{PropositionW} show that $A \mapsto Y_0(A)$ is a convex function, which reaches a maximum at $A=0$, with $Y_0(0)=0$. $Y_0(A)$ is therefore always above the straight line going through the origin with slope $A\,\frac{\partial\,Y}{\partial\,A}(0)$: $Y_0 \geq A\,Z_0(0)$, where $Z_0(0)$ is deduced from (\ref{Zeps0}). 

At $A \ll 1$, a Taylor expansion of $Y$ is written as follows
$$
Y_0(A)=Y_0(0)+A\,\frac{\textstyle \partial\,Y_0}{\textstyle \partial\,A}(0)+{\cal O}(A^2).
$$
If $A=0$, then $Y_0=Y=0$, and (\ref{Zeps0}) again provides $Z_0(0)=\frac{\partial\,Y_0}{\partial\,A}(0)$. 
\qquad\end{proof}\\

\begin{figure}[htbp]
\begin{center}
\begin{tabular}{cc}
(a)&(b)\\
\includegraphics[scale=0.31]{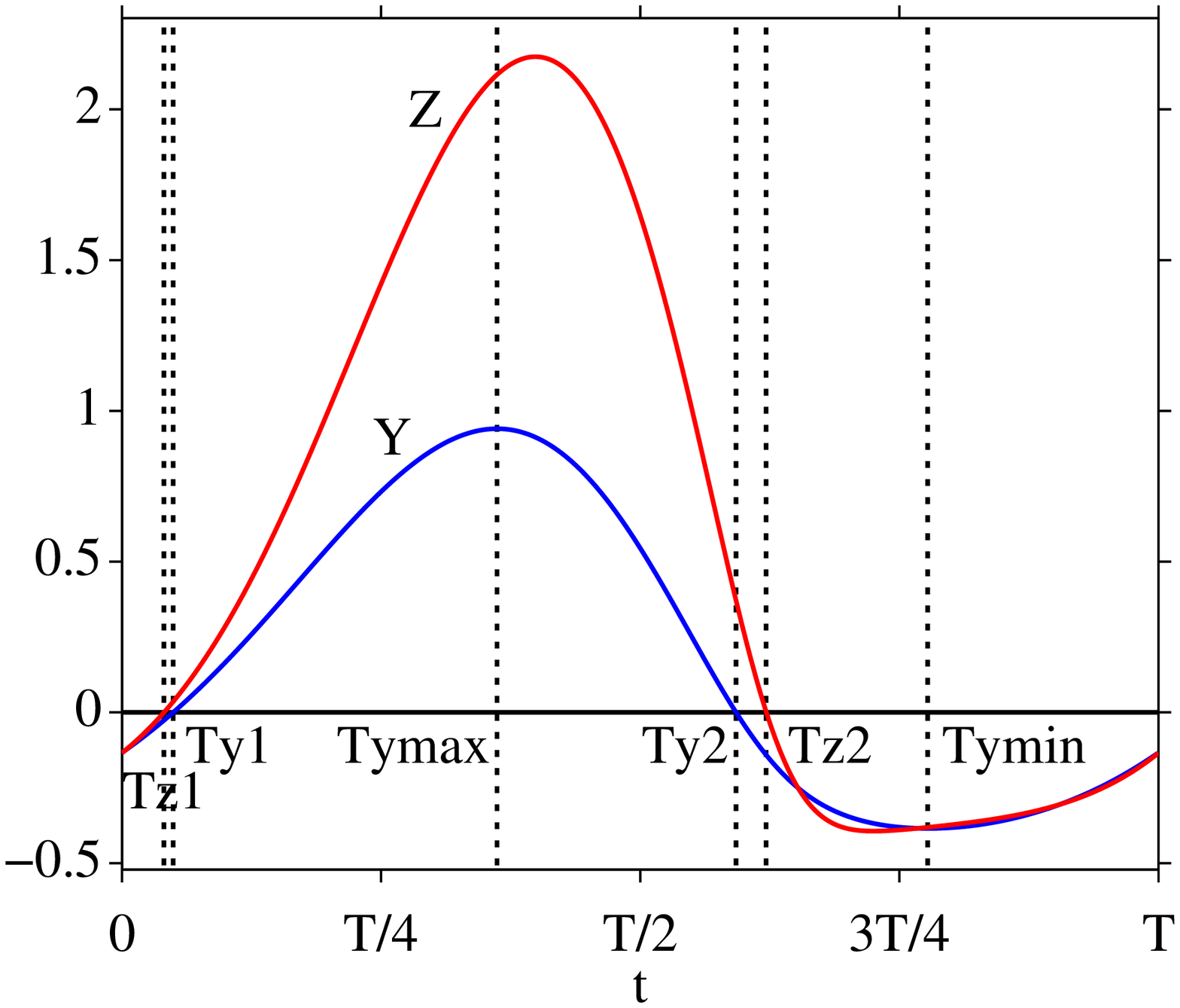}&
\includegraphics[scale=0.31]{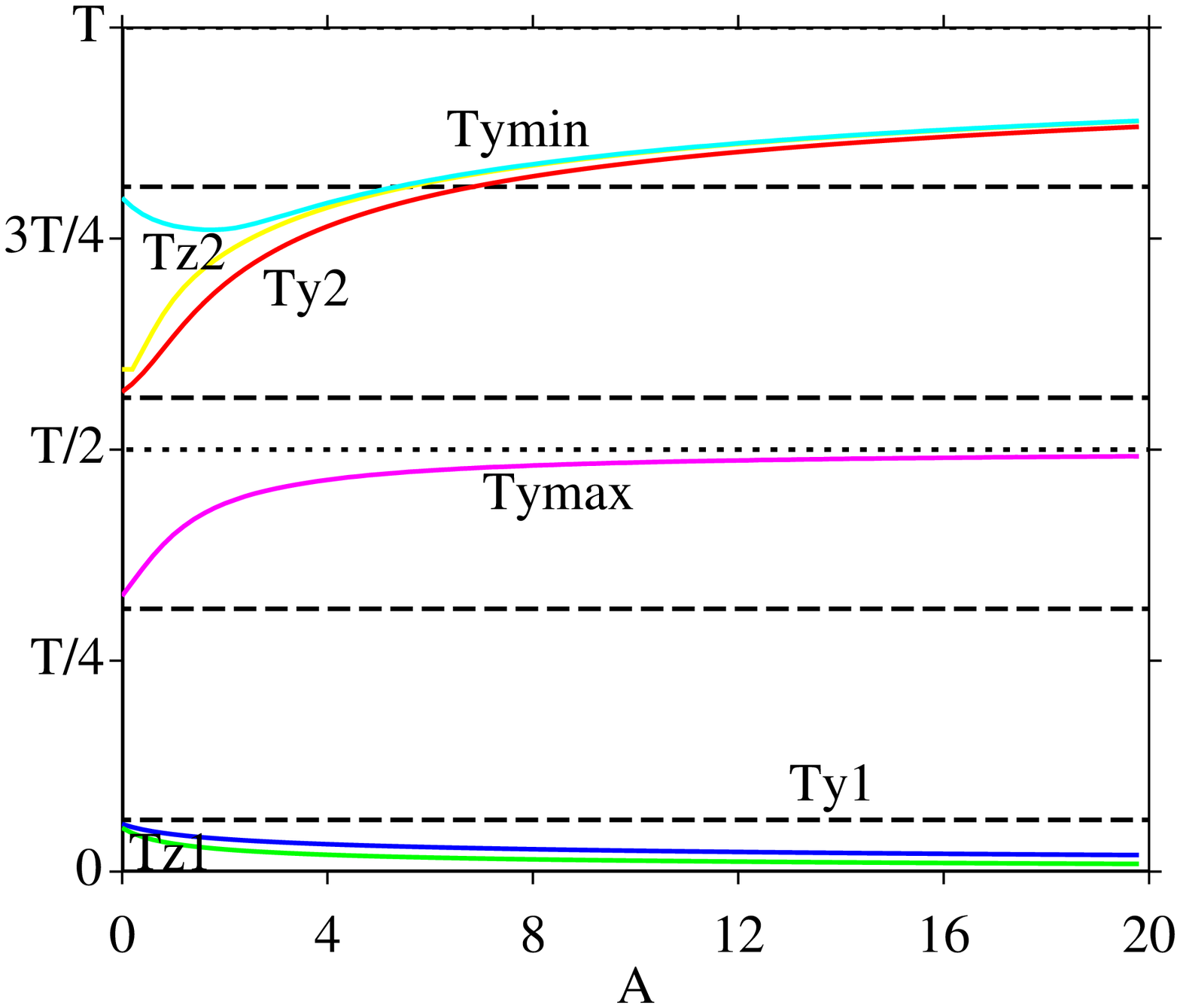}
\end{tabular}
\end{center}
\caption{Critical times in $Y$ and $Z$, in the case of model 1 (\ref{Model1}). Snapshot at $A=0.6$ (a): the vertical dotted lines denote successively $t_{Z_1}$, $t_{Y_1}$, $t_{Y_{\max}}$, $t_{Y_2}$, $t_{Z_2}$ and $t_{Y_{\min}}$.  Parametric study with $A$ (b): a horizontal dotted line has been added at $T/2$, and four horizontal dashed lines have been added at $\theta/\omega+i\,T/4$ ($i=0, ..., 3$) to show the locations of critical times for $A\rightarrow 0$. See proposition \ref{PropositionTy} and corollary \ref{CorollaryTy}.} 
\label{FigTyTz}
\end{figure}

The evolution of $Y_0$ with $A$ is shown in figure \ref{FigY0}. Critical times in $Y$ are defined and estimated in the following proposition and corollary. 

\begin{proposition}
In $[0,\,T]$, the $T$-periodic solution $Y$ of (\ref{ODE_Y}), which has two roots $t_{Y_1}$, $t_{Y_2}$, is maximum and positive at $t_{Y_{\max}}$, and minimum and negative at $t_{Y_{\min}}$. The critical points are ordered as follows: 
$$
\begin{array}{l}
0<t_{Y_1}<T/2<t_{Y_2}<T,\qquad  t_{Y_1}<t_{Y_{\max}}<t_{Y_2}<t_{Y_{\min}},\\     
[6pt]
 T/4<t_{Y_{\max}}<T/2,\qquad \hspace{0.8cm} 3\,T/4<t_{Y_{\min}}<T.
\end{array}
$$
Wit $A \ll 1$ and using $\theta$ defined in (\ref{Zeps0}), we also obtain the following local estimates
\begin{equation}
\left\|
\begin{array}{l}
\displaystyle
Y(t, \,A)=\frac{\textstyle A}{\textstyle \sqrt{1+\omega^2}}\,\sin\left(\omega t-\theta\right)+{\cal O}\left(A^2\right),\\
[10pt]%\\
\displaystyle
t_{Y_1}(0^+)=\frac{\textstyle \theta}{\textstyle \omega},\quad 
t_{Y_{\max}}(0^+)=\frac{\textstyle \theta}{\textstyle \omega}+\frac{\textstyle T}{\textstyle 4},\quad 
t_{Y_2}(0^+)=\frac{\textstyle \theta}{\textstyle \omega}+\frac{\textstyle T}{\textstyle 2},\quad
t_{Y_{\min}}(0^+)=\frac{\textstyle \theta}{\textstyle \omega}+\frac{\textstyle 3\,T}{\textstyle 4}.
\end{array}
\right.
\label{Yeps0}
\end{equation}
\label{PropositionTy}
\end{proposition}

\begin{proof}
The phase portrait of (\ref{ODE_Y}) and the proposition \ref{PropositionY0} mean that $Y(t)$ must cross the $t$-axis twice, at $t_{Y_1}$ and $t_{Y_2}$ in $]0,\,T[$. At $t\geq T/2$, the isocline $I_0^y$ (\ref{IsoclineY}) is negative (figure \ref{FigIsoclineY}). Consequently, $t_{Y_1}\geq T/2$ means that $Y$ crosses $I_0^y$ twice at negative values of $Y$, preventing $Y$ from crossing the $t$-axis, which is impossible.

The slope of $Y$ is null at $t_{Y_{\max}}$, which means that $Y$ crosses $I_0^y$. This occurs only when the slope of $I_0^y$  is negative, between $T/4$ and $3\,T/4$. Since $I_0^y$ must also be positive at the intersection point, it is possible only when $T/4<t_{Y_{\max}}<T/2$. $Y$ subsequently decreases but cannot cross the $t$-axis as long as $I_0^y$ positive, which means that $t_{Y_2}>T/2$.

Lastly, the slope of $Y$ is equal to zero at $t_{Y_{\min}}$, which means that $Y$ crosses $I_0^y$. This is possible only when the slope of $I_0^y$ is positive. Since $I_0^y$ must also be negative, this means that $t_{Y_{\min}}>3\,T/4$. 

To prove the local estimates, a Taylor expansion of $Y$ is written as follows
$$
\begin{array}{lll}
\displaystyle
Y(t,\,A)&=&
\displaystyle Y(t,\,0)+A\,\frac{\textstyle \partial\,Y}{\textstyle \partial\,A}(t,\,0)+{\cal O}(A^2),\\
[8pt]
&=& \displaystyle A\,Z(t,\,0)+{\cal O}(A^2).
\end{array}
$$
Equation (\ref{Zeps0}) and straightforward calculations complete the proof.
\qquad \end{proof}\\

\begin{corollary}
The following properties hold for all $A > 0$:
\begin{romannum}
\item $[t_{Y_1},\,t_{Y_2}]\subset [t_{Z_1},\,t_{Z_2}]$, or equivalently $Z(t_{Y_1},\,A)\geq 0$ and $Z(t_{Y_2},\,A)\geq 0$; \\
\item $\displaystyle \frac{\textstyle \partial\,t_{Y_1}}{\textstyle \partial\,A}<0, 
\quad \frac{\textstyle \partial\,t_{Y_2}}{\textstyle \partial\,A}>0$;\\
\item $\displaystyle t_{Y_1}<\frac{\textstyle \theta}{\textstyle \omega}<\frac{\textstyle T}{\textstyle 4}, 
\quad t_{Y_2}>\frac{\textstyle \theta}{\textstyle \omega}+\frac{\textstyle T}{\textstyle 2}$;\\
\item if $|f_{\min}|<\infty$, then 
$\displaystyle \lim_{A \rightarrow +\infty} t_{Y_{\max}}=\frac{\textstyle T}{\textstyle 2}$.
\end{romannum}
\label{CorollaryTy}
\end{corollary}

\begin{proof}

\underline{Property (i)}. Two subsets of the $t$-plane are defined:
$$
I(A)=\left\{t\in]0,\,T[;\,Y(t,\,A)>0\right\},\qquad
J(A)=\left\{t\in]0,\,T[;\,Z(t,\,A)>0\right\}.
$$
Equations (\ref{Zeps0}) and (\ref{Yeps0}) prove that $I(0^+)=J(0)$. Phase portraits of (\ref{ODE_Y}) and (\ref{ODE_Z}) show that $I$ and $J$ are open intervals:
$$
I(A)=\left]t_{Y_1},\,t_{Y_2}\right[,\qquad
J(A)=\left]t_{Z_1},\,t_{Z_2}\right[.
$$
Let us consider two forcing parameters $A_1$ and $A_2$, with $A_2>A_1>0$. Proposition \ref{PropositionW} gives $\frac{\partial \,Z}{\partial\,A}>0$, and hence $Z(t,\,A_1)<Z(t,\,A_2)$. Taking $t\in J(A_1)$ gives $0<Z(t,\,A_1)<Z(t,\,A_2)$, and therefore $t \in J(A_2)$: the interval $J$ increases strictly with $A$. To prove $I(A) \subset J(A)$, we take $t\in I(A)$, and hence $Y(t,\,A)>0$. Two cases can be distinguished:
\begin{itemize}
\item $t \notin I(0^+)\Rightarrow Y(t,\,0^+)<0$. The formula
$$
Y(t,\,A)=Y(t,\,0^+)+\int_{0^+}^A Z(t,\,s)\,ds
$$
and the monotonic increase in $J$ with $A$ yields $Z(t,\,A)>0$, hence $t \in J(A)$;
\item $t\in I(0^+)$. Arguments mentioned above yield $t \in J(0) \subset J(A)$, which completes the proof of (i).
\end{itemize} 

\underline{Property (ii)}. 
First we differentiate $Y(t_{Y_1},\,A)=0$ in terms of $A$:
$$
\frac{\textstyle \partial\,Y}{\textstyle \partial\,t}\,\frac{\textstyle \partial\,t_{Y_1}}{\textstyle \partial\,A}+\frac{\textstyle \partial\,Y}{\textstyle \partial\,A}=0.
$$
The time derivative is replaced via (\ref{ODE_Y}) 
$$
\left(f(Y(t_{Y_1},\,A))+A\,\sin \omega t_{Y_1}\right)\,\frac{\textstyle \partial\,t_{Y_1}}{\textstyle \partial\,A}+Z(t_{Y_1},\,A)=0.
$$
Since $f(0)=0$, we obtain
$$
\frac{\textstyle \partial\,t_{Y_1}}{\textstyle \partial\,A}=-\frac{\textstyle Z(t_{Y_1},\,A)}{\textstyle A\,\sin \omega t_{Y_1}}.
$$
Property (i) and proposition \ref{PropositionTy} then prove $\frac{\partial\,t_{Y_1}}{\partial\,A}<0$. A similar derivation gives
$$
\frac{\textstyle \partial\,t_{Y_2}}{\textstyle \partial\,A}=-\frac{\textstyle Z(t_{Y_2},\,A)}{\textstyle A\,\sin \omega t_{Y_2}}.
$$
Property (ii) and proposition \ref{PropositionTy} give $Z(t_{Y_2},\,A)\geq0$ and $\sin \omega t_{Y_2}<0$, and hence $\frac{\partial\,t_{Y_2}}{\partial\,A}>0$. 

\underline{Property (iii)}. The bounds of $t_{Y_1}$ and $t_{Y_2}$ follow from (\ref{Yeps0}) and property (ii). From the definition of $\theta$ (\ref{Zeps0}) and the property of $\arctan$, we also obtain
$$
\frac{\textstyle \theta}{\textstyle \omega}\leq\frac{\textstyle 1}{\textstyle \omega}\,\frac{\textstyle \pi}{\textstyle 2}=\frac{\textstyle T}{\textstyle 4},
$$
which provides an additional bound for $t_{Y_1}$.

\underline{Property (iv)}. If $A\geq|f_{\min}|$, $Y$ crosses the isocline $I_0^y$ at $t_{Y_{\max}}>T_2$: see (\ref{IsoclineY}), (\ref{AsymptotsY}) and figure \ref{IsoclineY}-(b). The bounds of proposition \ref{PropositionTy} ensure $T_2 <t_{Y_{\max}}<T/2$. Using $\displaystyle \lim_{A \rightarrow + \infty}T_2=T/2$ completes the proof.
\qquad\end{proof}\\

Numerical computations illustrating proposition \ref{PropositionTy} and corollary \ref{CorollaryTy} are shown in figure \ref{FigTyTz}. No theoretical results have been obtained about $\frac{\partial\,t_{Y_{\max}}}{\partial\,A}$ and $\frac{\partial\,t_{Y_{\min}}}{\partial\,A}$. The numerical simulations indicate that $t_{Y_{\max}}$ strictly increases with $A$, whereas $t_{Y_{\min}}$ first decreases up to a critical value, and then increases (b). Numerical simulations also indicate that $t_{Z_2}$ tends asymptotically towards $t_{Y_{\min}}$ at large values of $A$.

%-------------------------------------------------------------------------------------

\subsection{Mean dilatation of the crack}

\begin{figure}[htbp]
\begin{center}
\begin{tabular}{cc}
(a)&(b)\\
\includegraphics[scale=0.31]{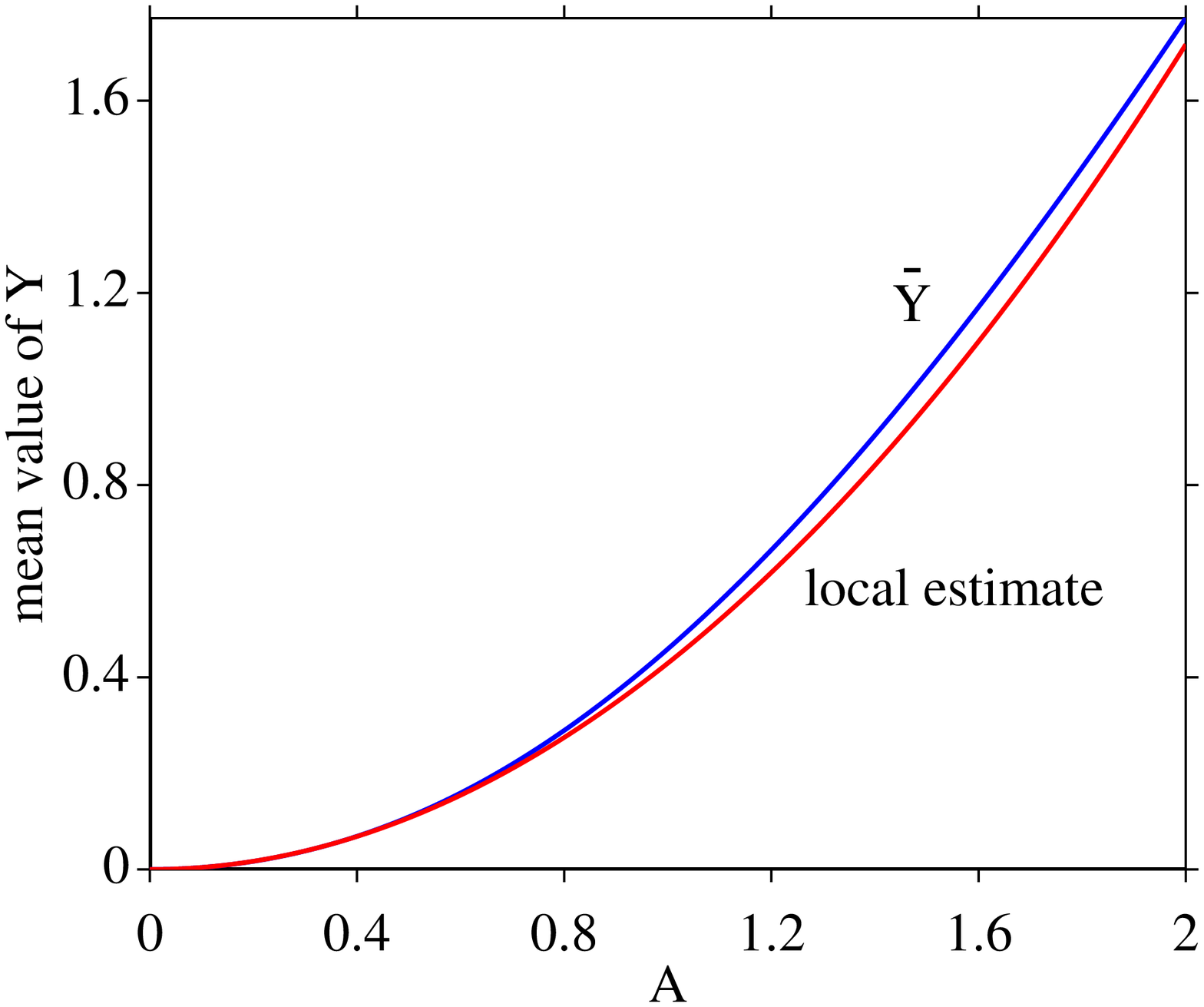}&
\includegraphics[scale=0.31]{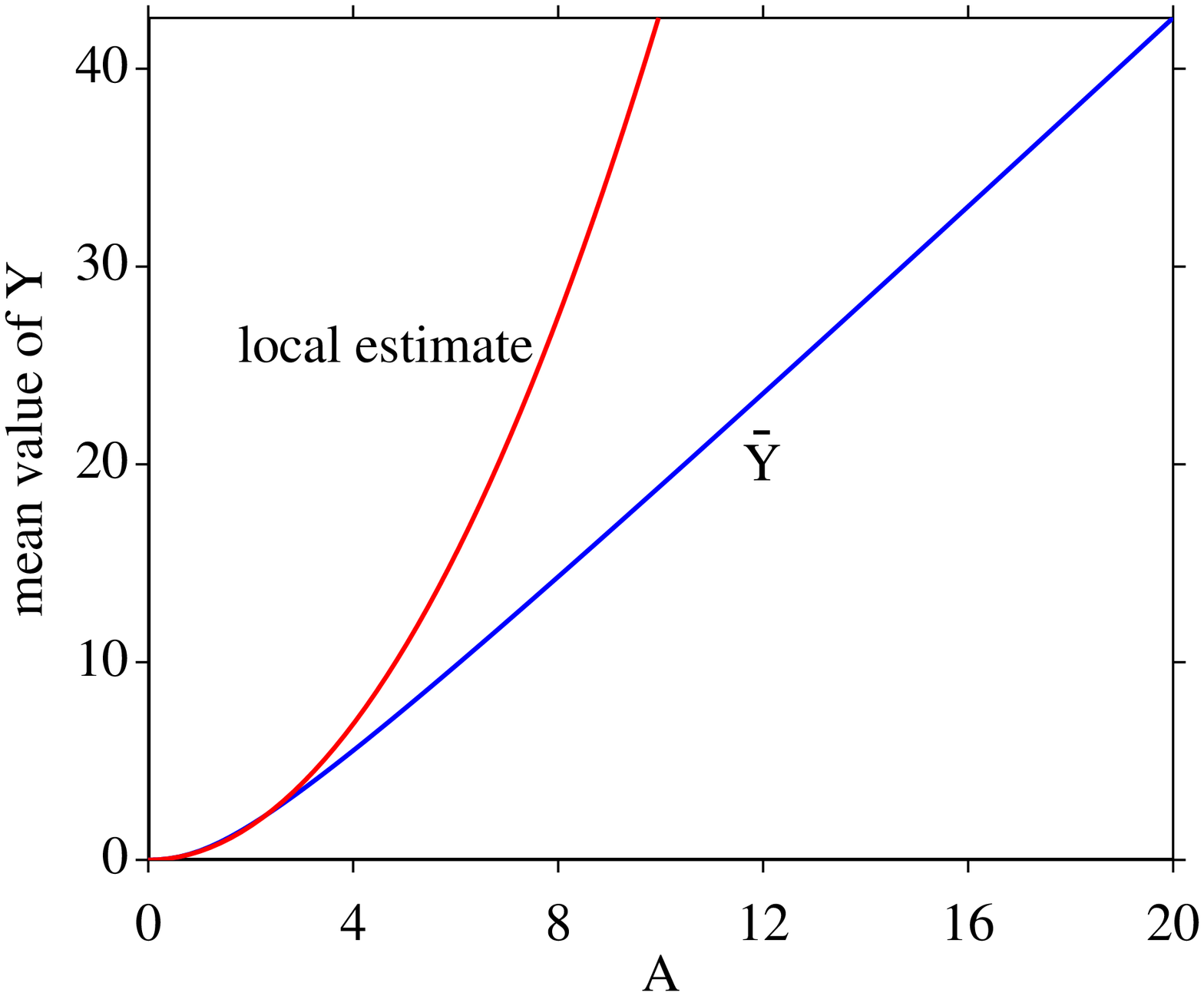}
\end{tabular}
\end{center}
\caption{Evolution of $\overline{Y}$ with the forcing amplitude $A$, at $0\leq A \leq 2$ (a) and $0 \leq A \leq 20$ (b), in the case of model 1 (\ref{Model1}). See theorem \ref{TheoremYbar}.} 
\label{FigYbar}
\end{figure}

As stated in the introduction, the main aim of this study was to prove that a jump in the mean elastic displacement occurs around the crack, as observed numerically. This jump amounts to a mean dilatation of the crack. The next theorem addresses this typically nonlinear phenomenon.

\begin{theorem}
The mean value of the $T$-periodic solution $Y$ in (\ref{ODE_Y}) is positive and increases strictly with the forcing amplitude:
$$
\overline{Y}>0,\qquad \frac{\textstyle \partial\,\overline{Y}}{\textstyle \partial\,A}>0.
$$
At small forcing levels, the following local estimate holds:
\begin{equation}
\overline{Y}=\frac{\textstyle f^{''}(0)}{\textstyle 4}\,\frac{\textstyle A^2}{\textstyle 1+\omega^2}+{\cal O}(A^4).
\label{YbarLocal}
\end{equation}
\label{TheoremYbar}
\end{theorem}

\begin{proof}
$T$-periodicity of $Y$ and $\sin \omega t$ in (\ref{ODE_Y}) yields $\overline{f(Y)}=0$. Applying Jensen's inequality to convex $f$ gives $f(\overline{Y})<\overline{f(Y)}=0$. Since $f(0)=0$ and $f^{'}<0$, then $\overline{Y}>0$. Proposition \ref{PropositionW} and (\ref{Zeps0}) yield $\overline{Z}(0)=0$ and $\overline{W}=\frac{\partial\,\overline{Z}}{\partial\,A}>0$, and hence $\overline{Z}(A)=\frac{\partial\,\overline{Y}}{\partial\,A}(A)>0$, which proves the second inequality. At $A\ll 1$, a Taylor expansion gives
$$
\begin{array}{lll}
\displaystyle
\overline{Y}(A) &=& \displaystyle \overline{Y}(0)+A\,\frac{\partial\,\overline{Y}}{\partial\,A}(0) +\frac{\textstyle A^2}{\textstyle 2}\,\frac{\partial^2\,\overline{Y}}{\partial\,A^2}(0)+\frac{\textstyle A^3}{\textstyle 6}\,\frac{\partial^3\,\overline{Y}}{\partial\,A^3}(0)+{\cal O}(A^4),\\
[6pt]
&=&\displaystyle \overline{Y}(0)+A\,\overline{Z}(0) +\frac{\textstyle A^2}{\textstyle 2}\,\overline{W}(0)+\frac{\textstyle A^3}{\textstyle 6}\,\overline{X}(0)+{\cal O}(A^4).
\end{array}
$$
From (\ref{Yeps0}) and (\ref{Zeps0}), we obtain $\overline{Y}(0)=0$ and $\overline{Z}(0)=0$. Lemmas \ref{LemmaW0} and \ref{LemmaX0} then give the local estimate of $\overline{Y}$. 
\qquad\end{proof}\\

The evolution of $\overline{Y}$ with $A$ is presented in figure \ref{FigYbar}. With $A=2$, the relative error between $\overline{Y}$ and its local estimate (\ref{YbarLocal}) is less than $5\%$. Figure \ref{FigYbar}-(a) may be rather misleading as far as moderate values of $A$ are concerned: it might seem to suggest that $\overline{Y}$ is always greater than the local estimate in (\ref{YbarLocal}). This is not in fact the case with greater values of $A$: the position of $\overline{Y}$ relative to its local estimate is not constant (b).

%-------------------------------------------------------------------------------------

\subsection{Maximum aperture of the crack}

Let 
$$
Y_{\max}=Y(t_{Y_{\max}})=\sup_{t\in[0,\,T]}Y(t)
$$
be the maximum value of the $T$-periodic solution $Y$. We introduce a  technical lemma.

\begin{lemma}
Let $h\in C^3\left(I\times\Lambda,\,\mathbb{R}\right)$, where $I$ and $\Lambda$ are open subsets of $\mathbb{R}$. The following properties are assumed to hold whatever the value of $\lambda$ in $\Lambda$:
\begin{romannum}
\item $\displaystyle h^+(\lambda)=\sup_{z \in I}h(z,\,\lambda)$ is reached at a single point $z^+(\lambda)\in I$;
\item $\displaystyle \frac{\textstyle \partial^2\,h}{\textstyle \partial\,z^2}(z^+(\lambda),\,\lambda)<0$;
\item $\forall z\in I$, $\displaystyle \frac{\textstyle \partial\,h}{\textstyle \partial\,\lambda}>0$, $\displaystyle \frac{\textstyle \partial^2\,h}{\textstyle \partial\,\lambda^2}>0$.
\end{romannum}
Then $h^+(\lambda)\in C^2(I,\,\mathbb{R})$, and $\lambda\mapsto h^+(\lambda)$ is a strictly increasing and convex function.
\label{LemmaYmax}
\end{lemma}

\begin{proof}
The definition of $z^+$ in (i) gives
\begin{equation}
\frac{\textstyle \partial\,h}{\textstyle \partial\,z}\left(z^+(\lambda),\,\lambda\right)=0.
\label{ProofLemmaYmax1}
\end{equation}
In addition, (ii) ensures that $\frac{\partial^2\,h}{\partial\,z^2}(z^+(\lambda),\,\lambda)\neq0$. The implicit function theorem can be applied to $\frac{\partial\,h}{\partial\,z}\in C^2(I\times\Lambda,\,\mathbb{R})$: therefore, $\lambda\mapsto z^+(\lambda) \in C^2(\Lambda,\,\mathbb{R})$, and $h^+(\lambda)=h(z^+,\,\lambda)\in C^2(I,\,\mathbb{R})$.

Based on the first property in (iii) and (\ref{ProofLemmaYmax1}),
$$
\begin{array}{lll}
\displaystyle
\frac{\textstyle d\,h^+}{\textstyle d\,\lambda}(\lambda) &=& \displaystyle \frac{\textstyle \partial\,h}{\textstyle \partial\,\lambda}(z^+,\,\lambda)+\frac{\textstyle \partial\,h}{\textstyle \partial\,z}(z^+,\,\lambda)\,\frac{\textstyle d\,z^+}{\textstyle d\,\lambda},\\
[10pt]
\displaystyle
&=& \displaystyle \frac{\textstyle \partial\,h}{\textstyle \partial\,\lambda}(z^+,\,\lambda)>0,
\end{array}
$$
which proves that $h^+$ is a strictly increasing function of $\lambda$. On the other hand, differentiating (\ref{ProofLemmaYmax1}) in terms of $\lambda$ gives
\begin{equation}
\frac{\textstyle \partial^2\,h}{\textstyle \partial\,\lambda\,\partial\,z}(z^+,\,\lambda)+\frac{\textstyle \partial^2\,h}{\textstyle \partial\,z^2}(z^+,\,\lambda)\,\frac{\textstyle d\,z^+}{\textstyle d\,\lambda}=0.
\label{ProofLemmaYmax2}
\end{equation}
Using (ii), the second property in (iii), and (\ref{ProofLemmaYmax2}), we obtain
$$
\begin{array}{lll}
\displaystyle
\frac{\textstyle d^2\,h^+}{\textstyle d\,\lambda^2}(\lambda) &=& \displaystyle \frac{\textstyle \partial^2\,h}{\textstyle \partial\,\lambda^2}(z^+,\,\lambda)+\frac{\textstyle \partial^2\,h}{\textstyle \partial\,z\,\partial\,\lambda}(z^+,\,\lambda)\,\frac{\textstyle d\,z^+}{\textstyle d\,\lambda},\\
[10pt]
&=& \displaystyle \frac{\textstyle \partial^2\,h}{\textstyle \partial\,\lambda^2}(z^+,\,\lambda)-\frac{\textstyle \partial^2\,h}{\textstyle \partial\,z^2}(z^+,\,\lambda)\left(\frac{\textstyle d\,z^+}{\textstyle d\,\lambda}\right)^2>0,
\end{array}
$$
which proves that $h^+$ is a strictly convex function of $\lambda$.
\qquad\end{proof}\\

\begin{theorem}
$Y_{\max}$ is a strictly increasing convex function of $A$, and the following bounds hold for all $A$:
\begin{equation}
\max\left(\frac{\textstyle A}{\textstyle \sqrt{1+\omega^2}},\,\frac{\textstyle A}{\textstyle \omega}\,\frac{\textstyle 1}{\textstyle \sqrt{1+\omega^2}}-\frac{\textstyle T}{\textstyle 2}\,|f_{\min}|\right)<Y_{\max}<\frac{\textstyle 2\,A}{\textstyle \omega}.
\label{Ymax}
\end{equation}
The first lower bound of $Y_{\max}$ is also a local estimate for small forcing levels $A\ll 1$
\begin{equation}
Y_{\max}=\frac{\textstyle A}{\textstyle \sqrt{1+\omega^2}}+{\cal O}\left(A^2\right).
\label{YmaxLocal}
\end{equation}
\label{TheoremYmax}
\end{theorem}

\begin{figure}[htbp]
\begin{center}
\begin{tabular}{cc}
(a)&(b)\\
\includegraphics[scale=0.31]{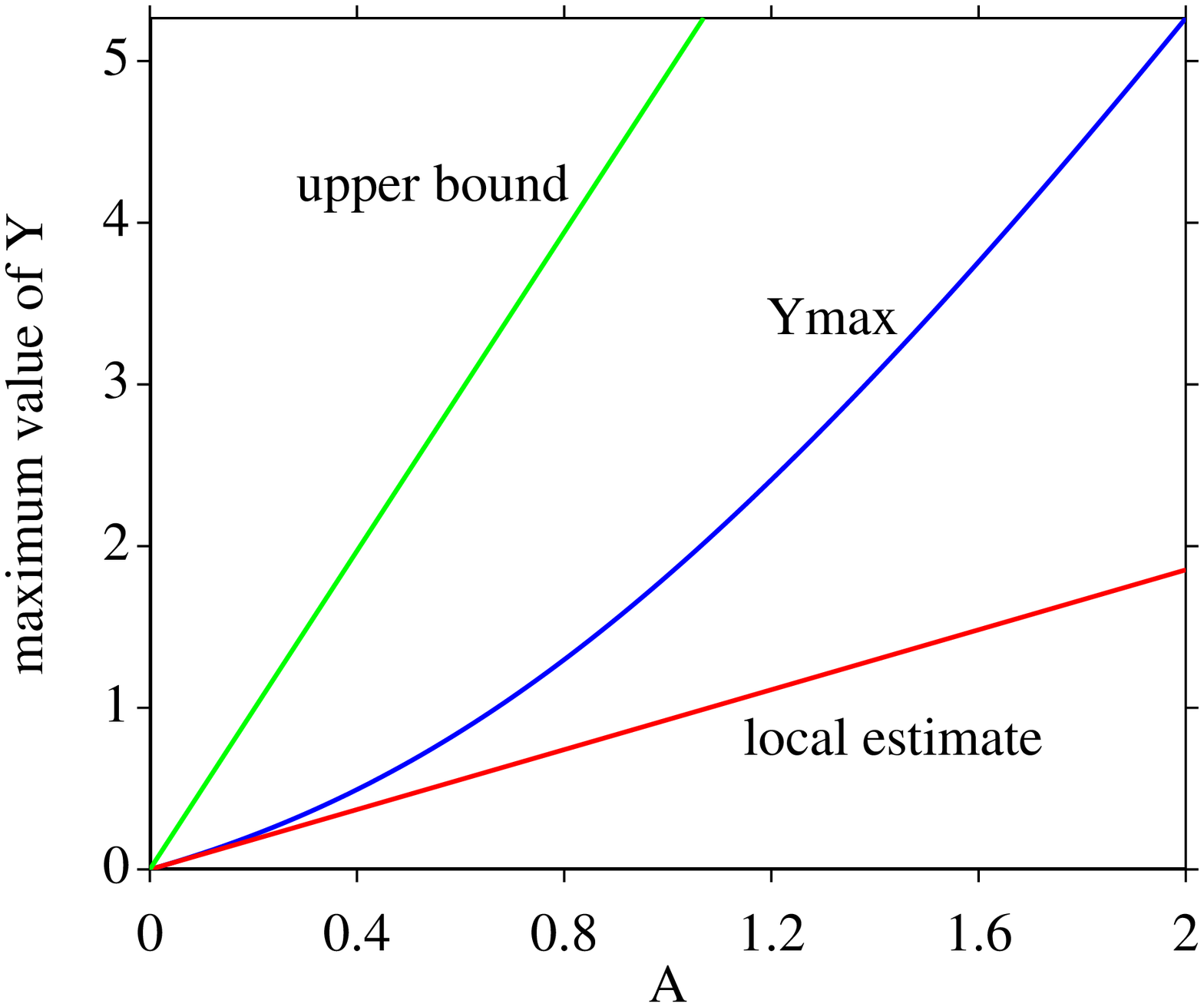} &
\includegraphics[scale=0.31]{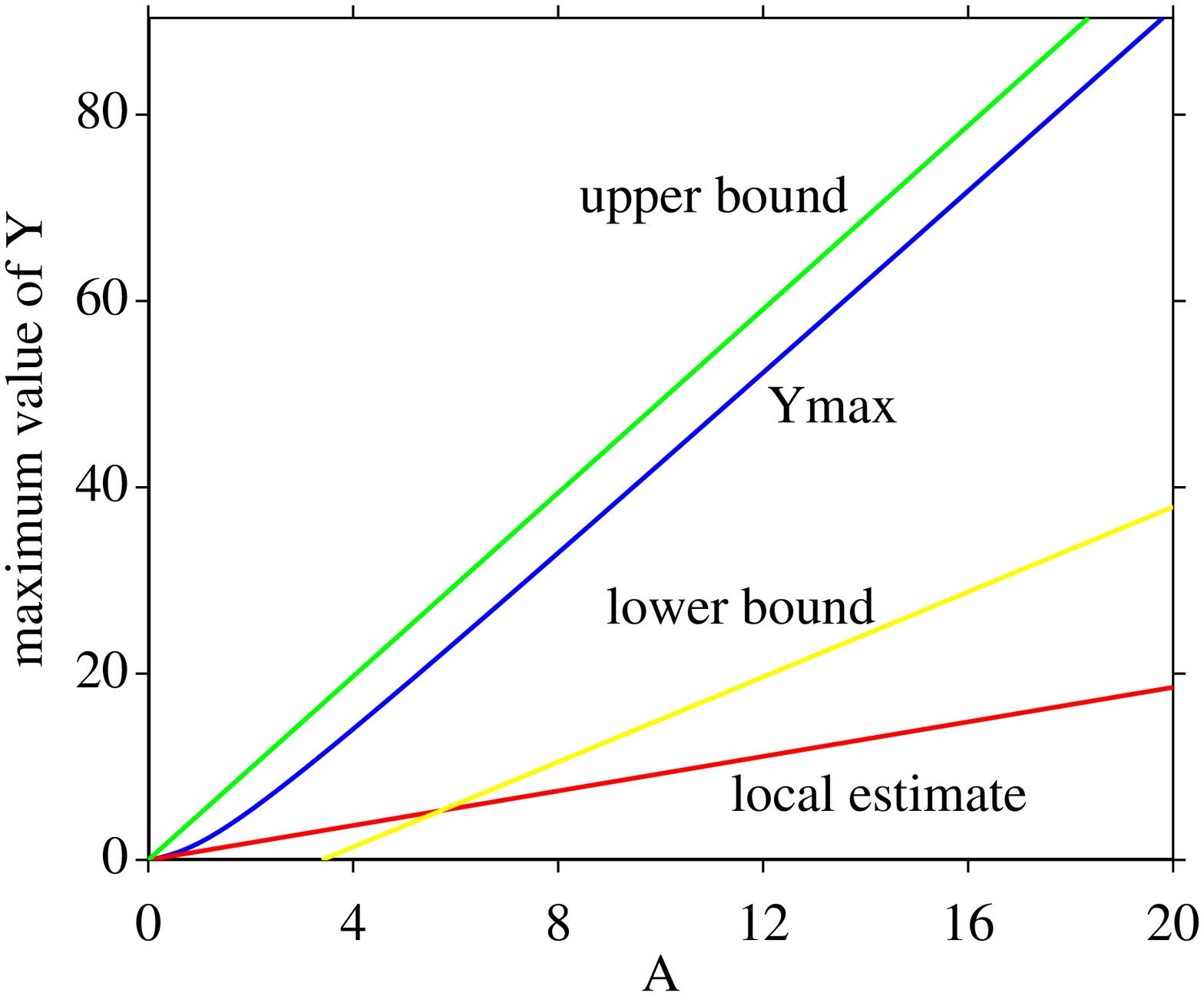} 
\end{tabular}
\end{center}
\caption{Evolution of $Y_{\max}$ with the forcing amplitude $A$, at $0\leq A\leq 2$ (a) and $0\leq A\leq 20$ (b), in the case of model 1 (\ref{Model1}). See theorem \ref{TheoremYmax}.} 
\label{FigYmax}
\end{figure}

\begin{proof}
To analyse the evolution of $Y_{\max}$ in terms of $A$, we proceed in three steps. First, the phase portrait of $Y$ and the proposition \ref{PropositionTy} show that $Y(t,\,A)$ reaches its maximum value at a single time between $T/4$ and $T/2$. Secondly, the definition of $t_{Y_{\max}}$ and proposition \ref{PropositionTy} give
$$
\begin{array}{lll}
\displaystyle
\frac{\textstyle \partial^2\,Y}{\textstyle \partial\,t^2}(t_{Y_{\max}},\,A)&=& \displaystyle f^{'}(Y_{\max})\,\frac  {\textstyle \partial\,Y}{\textstyle \partial\,t}(t_{Y_{\max}},\,A)+A\,\omega\,\cos \omega t_{Y_ {\max}},\\
[10pt]
&=& \displaystyle A\,\omega\,\cos \omega t_{Y_{\max}}<0.
\end{array}
$$
Thirdly, corollary \ref{CorollaryTy} and proposition \ref{PropositionW} yield $\frac{\partial\,Y}{\partial\,A}>0$ and $\frac{\partial\,Y^2}{\partial\,A^2}>0$ on $[T/4,\,T/2]$. The three hypotheses in lemma \ref{LemmaYmax} are therefore satisfied, which proves the convex increasing of $Y_{\max}$ with $A$.

The bounds of $Y_{\max}$ are proved by building upper and lower solutions of $Y$, as in the proof of proposition \ref{PropositionYperiod}, case 3. In $]t_{Y_1},\,t_{Y_{\max}}[$, proposition \ref{PropositionTy} ensures $Y> 0$, which means that $f_{\min}< f(Y)<0$, and hence $f_{\min}+A\,\sin \omega t < \frac{d\,Y}{d\,t}< A\,\sin \omega t$. Integrating this inequality from $t_{Y_1}$ to $t_{Y_{\max}}$ gives
$$
\frac{\textstyle A}{\textstyle \omega}\,\left(\cos \omega t_{Y_1}-\cos \omega t_{Y_{\max}}\right)+\left(t_{Y_1}-t_{Y_{\max}}\right)\,|f_{\min}|< Y_{\max}< \frac{\textstyle A}{\textstyle \omega}\,\left(\cos \omega t_{Y_1}-\cos \omega t_{Y_{\max}}\right).
$$
Proposition \ref{PropositionTy} and corollary \ref{CorollaryTy}-(iii) give the bounds $T/4<t_{Y_{\max}}<T/2$ and $0<t_{Y_1}<\theta/\omega$, and hence 
$$
\left\{
\begin{array}{l}
\displaystyle
\cos \theta < \cos \omega t_{Y_1} < 1,\\
[6pt]
\displaystyle
0 < - \cos \omega t_{Y_{\max}}<1.
\end{array}
\right.
$$
These inequalities, together with the definition of $\theta$ in (\ref{Zeps0}), yield
$$
\frac{\textstyle A}{\textstyle \omega}\,\frac{\textstyle 1}{\textstyle \sqrt{1+\omega^2}}<\frac{\textstyle A}{\textstyle \omega}\,\left(\cos \omega t_{Y_1}-\cos \omega t_{Y_{\max}}\right)<\frac{\textstyle 2\,A}{\textstyle \omega}.
$$
Proposition \ref{PropositionTy} also gives $t_{Y_1}-t_{Y_{\max}}>-T/2$, which gives the upper bound and a first lower bound in (\ref{Ymax}). The latter lower bound is not always an optimum bound: when $A < \frac{T}{2}\,|f_{\min}|\,\omega\,\sqrt{1+\omega^2}$, it is negative, whereas proposition \ref{PropositionTy} states that $Y_{\max}>0$ at all values of $A$. We therefore take advantage of the convex increasing of $Y_{\max}$ with $A$, which was previously proved. Since $Y_{\max}(0)=0$, then $Y_{\max}(A)$ is always above the straight line going through the origin with the slope $A\,Z(t_{Y_{\max}}(0^+),\,0)$. Proposition \ref{PropositionTy} and equation (\ref{Zeps0}) therefore give a second lower bound in (\ref{Ymax}). A second-order Taylor expansion of $Y_{\max}(A)$ proves that this lower bound is also the local estimate of $Y_{\max}$ when $A\ll1$.
\qquad\end{proof}\\

The evolution of $Y_{\max}$ stated in theorem \ref{TheoremYmax} is shown in figure \ref{FigYmax}. If $|f_{\min}|<\infty$ and $\omega\neq 1$, the two lower bounds given in (\ref{Ymax}) intersect at $A=\frac{T}{2}\,|f_{\min}|\,\frac{\omega}{1-\omega}\,\sqrt{1+\omega^2}$. At higher values of $A$, the second lower bound is more accurate than the first one.

%-------------------------------------------------------------------------------------

\subsection{Maximum closure of the crack}

\begin{figure}[htbp]
\begin{center}
\begin{tabular}{cc}
(a)&(b)\\
\includegraphics[scale=0.31]{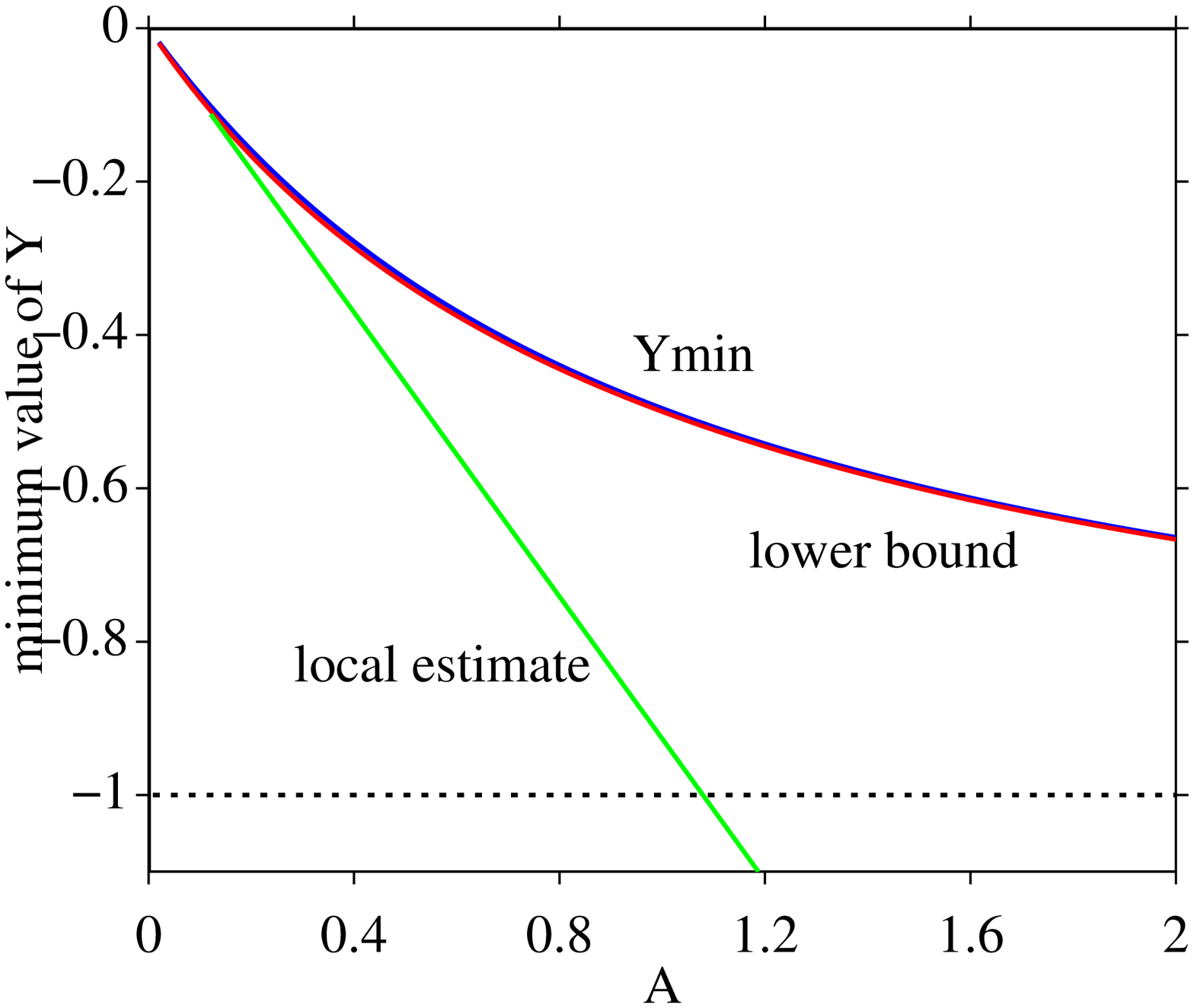} &
\includegraphics[scale=0.31]{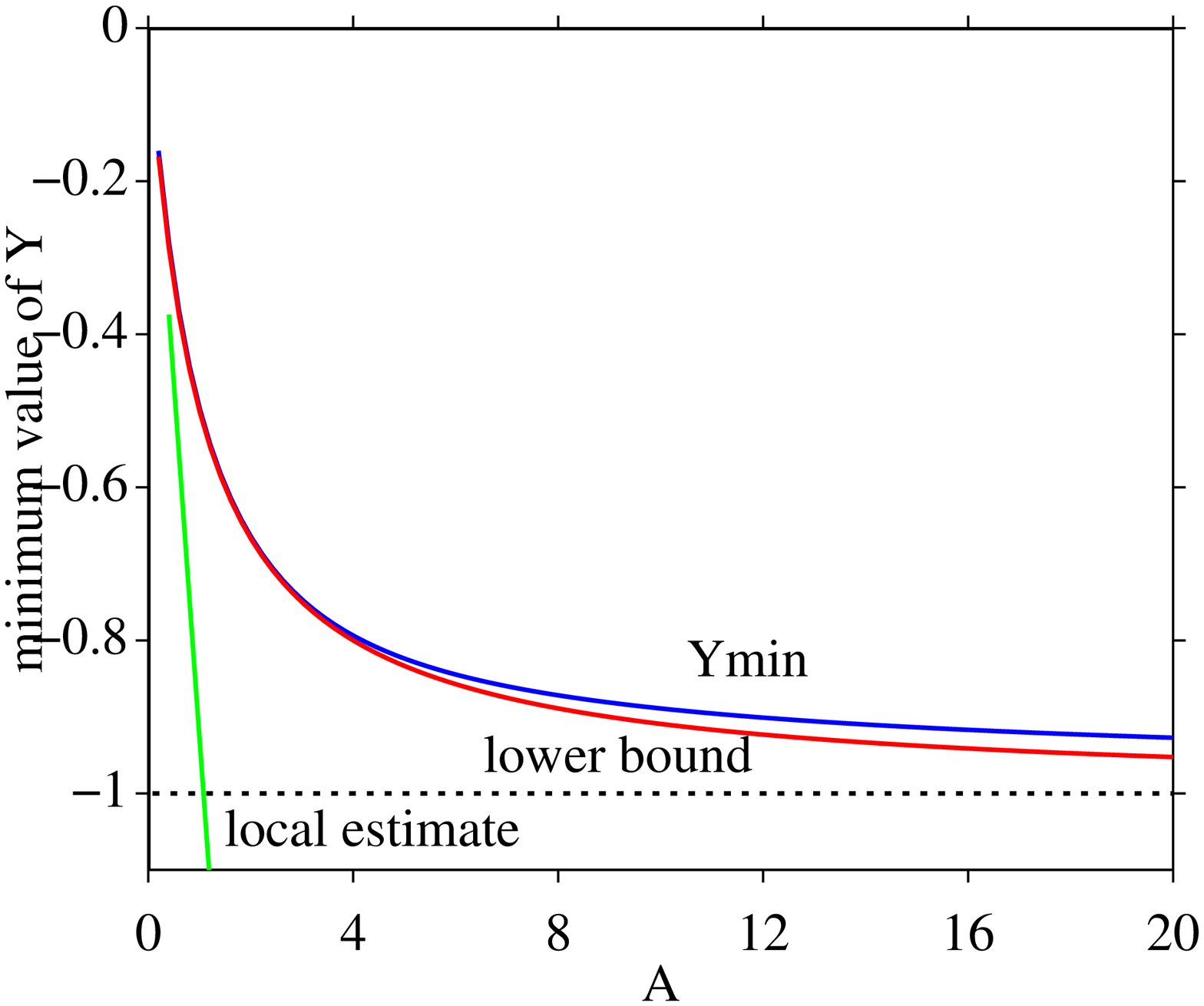} 
\end{tabular}
\end{center}
\caption{Evolution of $Y_{\min}$ with the forcing amplitude $A$, with $0\leq A\leq 2$ (a) and $0\leq A\leq 20$ (b), in the case of model 1 (\ref{Model1}). See theorem \ref{TheoremYmin}.} 
\label{FigYmin}
\end{figure}

Let
$$
Y_{\min}=Y(t_{Y_{\min}})=\inf_{t\in[0,\,T]}Y(t)
$$
be the minimum value of the $T$-periodic solution $Y$. 

\begin{theorem}
The following properties are satisfied by $Y_{\min}$:
\begin{romannum}
\item for all $A$, we obtain the bounds:
\begin{equation}
-1<f^{-1}(A)\leq Y_{\min} \leq 0;
\label{Ymin}
\end{equation}
\item at small forcing levels, we also obtain the local estimate
\begin{equation}
Y_{\min}=-\frac{\textstyle A}{\textstyle \sqrt{1+\omega^2}}+{\cal O}\left(A^2\right);
\label{YminLocal}
\end{equation}
\item at infinite forcing levels, the crack is completely closed:
$$
\lim_{A \rightarrow + \infty}Y_{\min}=-1.
$$
\end{romannum}
\label{TheoremYmin}
\end{theorem}

\begin{proof}

\underline{Property (i)}.
Proposition \ref{PropositionTy} shows that $Y_{\min}$ is negative. As in the proof of proposition \ref{PropositionY0}, the existence of the lower nonporous fence $f^{-1}(A)$ means that $Y_{\min}\geq f^{-1}(A)$. 

\underline{Property (ii)}. As in the proofs of proposition \ref{PropositionY0} and theorem \ref{TheoremYmax}, the local estimate when $A \ll 1$ follows from a second-order Taylor expansion of $Y$ and from (\ref{Zeps0}).

\underline{Property (iii)}. First we consider the property $\overline{f(Y)}=0$ used in the proof of theorem \ref{TheoremYbar}. Evolution of $f$ in (\ref{ODE_Y}) and the definition of $t_{Y_1}$, $t_{Y_2}$ in proposition \ref{PropositionTy} give
$$
\begin{array}{lll}
0 &=& \displaystyle
\int_0^{t_{Y_1}}f(Y(\tau))\,d\tau+\int_{t_{Y_1}}^{t_{Y_2}}f(Y(\tau))\,d\tau+\int_{t_{Y_2}}^Tf(Y(\tau))\,d\tau,\\
[8pt]
&=& \displaystyle -\left(\int_0^{t_{Y_1}}|f(Y(\tau))|\,d\tau+\int_{t_{Y_2}}^T|f(Y(\tau))|\,d\tau\right)+\int_{t_{Y_1}}^{t_{Y_2}}|f(Y(\tau))|\,d\tau,
\end{array}
$$
hence
\begin{equation}
\underbrace{\left(\int_0^{t_{Y_1}}|f(Y(\tau))|\,d\tau+\int_{t_{Y_2}}^T|f(Y(\tau))|\,d\tau\right)}_{I}=\underbrace{\int_{t_{Y_1}}^{t_{Y_2}}|f(Y(\tau))|\,d\tau}_{J}.
\label{ProofYminIJ}
\end{equation}
Secondly, let us consider the ODE
\begin{equation}
\left\{
\begin{array}{l}
\displaystyle
\frac{\textstyle d\,s}{\textstyle d\,t}=-s+A\,\sin\omega t,\\
[8pt]
\displaystyle
s(t_{Y_1})=0.
\end{array}
\right.
\label{ProofYminODE}
\end{equation}
Straightforward calculations and property (iii) of corollary \ref{CorollaryTy} give
\begin{equation}
s(t)=\frac{\textstyle A}{\textstyle \sqrt{1+\omega^2}}\,\left(|\sin (\omega t_{Y_1}-\theta)|\,e^{-(t-t_{Y_1})} +\sin(\omega t-\theta)\right),
\label{ProofYminS}
\end{equation}
with $\theta$ (\ref{Zeps0}). Comparison between (\ref{ODE_Y}) and (\ref{ProofYminODE}) shows that $s$ is a lower solution of $Y$ when $Y\geq0$: $\forall\,t\in[t_{Y_1},\,t_{Y_2}]$, $Y(t)\geq s(t)$. Since $t_{Y_{\max}}<t_{Y_2}$, it follows from (\ref{ProofYminIJ}) and (\ref{ProofYminS}) that
$$
\begin{array}{lll}
J&\geq& \displaystyle \int_{t_{Y_1}}^{t_{Y_{\max}}}|f(Y(\tau))|\,d\tau \geq \int_{t_{Y_1}}^{t_{Y_{\max}}}|f(s(\tau))|\,d\tau 
,\\
[8pt]
&\geq & \displaystyle \frac{\textstyle A}{\textstyle \sqrt{1+\omega^2}} \int_{t_{Y_1}}^{t_{Y_{\max}}} \left(|\sin (\omega t_{Y_1}-\theta)|\,e^{-(\tau-t_{Y_1})} +\sin(\omega \tau-\theta)\right)\,d\tau,\\
[9pt]
&\geq & \displaystyle \frac{\textstyle A}{\textstyle \sqrt{1+\omega^2}} \,\Delta,
\end{array}
$$
with
$$
\Delta=|\sin (\omega t_{Y_1}-\theta)|\,\left(1-e^{(t_{Y_{\max}}-t_{Y_1})}\right)+\frac{\textstyle 1}{\textstyle \omega}\left(\cos(\omega t_{Y_1}-\theta)-\cos(\omega t_{Y_{\max}}-\theta)\right).
$$
Proposition \ref{PropositionTy} and property (iii) of corollary \ref{CorollaryTy} give
$$
\left\{
\begin{array}{l}
\displaystyle
1-e^{T/2}<1-e^{(t_{Y_{\max}}-t_{Y_1})}<1-e^{T/4-\theta/\omega},\\
[6pt]
\displaystyle
0<|\sin (\omega t_{Y_1}-\theta)|<\sin \theta,\\
[6pt]
\displaystyle
\frac{\textstyle 2\,\cos\theta}{\textstyle \omega}< \frac{\textstyle 1}{\textstyle \omega}\left(\cos(\omega t_{Y_1}-\theta)-\cos(\omega t_{Y_{\max}}-\theta)\right) <\frac{\textstyle 1-\sin\theta}{\textstyle \omega},
\end{array}
\right.
$$
hence $\Delta$ is bounded independently of $A$. Finally, we obtain
$$
J\geq \frac{\textstyle 2\,A}{\textstyle \omega\,\left(1+\omega^2\right)}.
$$
$J$ blows up when $A \rightarrow +\infty$. Equation (\ref{ProofYminIJ}) means that $I$ behaves in a similar way. Since $Y<0$ in $]0,\,t_{Y_1}[\cup]t_{Y_2},\,T[$ and given $f$ in (\ref{ODE_Y}), $Y$ must tend towards -1 when $A \rightarrow +\infty$, to make $I$ blow up, which concludes the proof.
\qquad\end{proof}\\

The evolution of $Y_{\min}$ stated in theorem \ref{TheoremYmin} is shown in figure \ref{FigYmin}. At moderate forcing levels (a), one cannot distinguish between $Y_{\min}$ and its lower bound. At higher forcing levels (b), one observes that $Y_{\min}$ is above its lower bound, as stated in theorem \ref{TheoremYmin}. Although no rigorous proof has been obtained so far, the numerical simulations indicate that $Y_{\min}$ decreases strictly as $A$ increases. 

%\vspace{1.5cm}
%-------------------------------------------------------------------------------------
%-------------------------------------------------------------------------------------

\section{Generalization}\label{SecGeneralization}

\subsection{Periodic forcing}\label{SecGenePeriodic}

Most of the results obtained in sections \ref{SecPreRes} and \ref{SecQualitY} were based on the simple analytical expression of sinusoidal forcing. In the case of more general forcing, the key properties stated in proposition \ref{PropositionTy} and corollary \ref{CorollaryTy} are lost, which makes it impossible to obtain estimates such as theorems \ref{TheoremYmax} and \ref{TheoremYmin}. However, some properties are maintained if two assumptions are made about $S$ in (\ref{ODEorig}): $T_0$-periodicity, and a null mean value. The latter assumption is physically meaningfull: as deduced from (\ref{Vinc}), a non-null mean value of $S$ results in an amplitude of $u_I$ that increases linearly with $t$.

Since the techniques required are the same as those used in the above sections, the derivations will be shortened. Setting the parameters and the Fourier decomposition of the source as follows
%\begin{equation}
%\begin{array}{l}
%\displaystyle
%A=\frac{\textstyle 1}{\textstyle \beta\,d\,\rho_0\,c_0^2}\,\max_{t\in[0,\,T]}\left|S(t/\beta)\right|,\quad T=\beta\,T_0,%\quad \omega=\frac{\textstyle 2\,\pi}{\textstyle T},\quad f(y)=-{\cal F}(y),\\
%[10pt]
%\displaystyle
%g(t)=\frac{\textstyle S(t/\beta)}{\textstyle \displaystyle \max_{t\in [0,\,T]}\left|S(t/\beta)\right|}=\sum_{n=1}^ {\infty}\left(a_n\,\sin n\omega t+b_n\,\cos n\omega t\right), 
%\label{AdimGene}
%\end{array}
%\end{equation}
\begin{equation}
\begin{array}{l}
\displaystyle
v_0=\frac{\textstyle 1}{\textstyle 2\,\rho_0\,c_0^2}\,\max_{t\in[0,\,T]}\left|S(t/\beta)\right|,
\quad A=\frac{\textstyle 2\,v_0}{\textstyle \beta\,d},\quad y=\frac{\textstyle [u(\alpha,\,t)]}{\textstyle d},\quad f(y)=-{\cal F}(y),\\
[10pt]
\displaystyle
g(t)=\frac{\textstyle 1}{\textstyle 2\,\rho_0\,c_0^2\,v_0}\,\left|S(t/\beta)\right|=\sum_{n=1}^{\infty} \left(a_n\,\sin n\omega t+b_n\,\cos n\omega t\right), 
\quad T=\beta\,T_0,\quad \omega=\frac{\textstyle 2\,\pi}{\textstyle T},
\label{AdimGene}
\end{array}
\end{equation}
we obtain the model problem
\begin{equation}
\left\{
\begin{array}{l}
\displaystyle
\frac{\textstyle d\,y}{\textstyle d\,t}=f(y)+A\,g(t)=F(t,\,y),\\
[6pt]
\displaystyle
f:\,]-1,\,+\infty[\rightarrow ]f_{\min},\,+\infty[,\quad \lim_{y\rightarrow-1}f(y)=+\infty,\quad -\infty\leq f_{\min}<0,\\
[6pt]
\displaystyle
f(0)=0,\quad f^{'}(0)=-1,\quad f^{'}(y)<0<f^{''}(y),\\
[6pt]
\displaystyle
g(t+T)=g(t),\quad \overline{g}=0,\quad \left|g(t)\right|\leq 1,\quad \max_{t\in [0,\,T]}\left|g(t)\right|=1,\quad 0\leq A<+\infty,\\
[6pt]
\displaystyle
y(0)=y_0\in]-1,\,+\infty[.
\end{array}
\right.
\label{ODE_Ygene}
\end{equation}
The isocline of zero slope of (\ref{ODE_Ygene}) is
\begin{equation}
I_0^y(t)=f^{-1}(-A\,g(t)).
\label{IsoclineYgene}
\end{equation} 

%-------------------------------------------------------------------------------------

\subsection{Periodic solution}\label{SecPeriodic}

The following result generalizes the proposition \ref{PropositionYperiod}: since the same notations and processes are used, the proof is only sketched here.

\begin{proposition}
There is a unique $T$-periodic solution $Y(t)$ of (\ref{ODE_Ygene}). This solution is asymptotically stable.
\label{PropositionYgene}
\end{proposition}

\begin{proof} Three cases can be distinguished.

\underline{Case 1:} $A=0$. This case is identical to case 1 in proposition \ref{PropositionYperiod}.

\underline{Case 2:} $0<A<|f_{\min}|$. Based on the funnel's theorem, the compact set
$$
K_0=\left[f^{-1}(A),\,f^{-1}(-A)\right]
$$
is invariant under the flow of (\ref{ODE_Ygene}), and hence $\Pi(K_0)\subset K_0$. A fixed point argument shows the existence of a $T$-periodic solution $Y$. Equations (\ref{ProofYperiodUnique}) and (\ref{ProofYperiodA0}) still hold: the uniqueness and attractivity of $Y$ are straightforward consequences of $f^{'}(y)<0$. 

\underline{Case 3:} $A\geq |f_{\min}|$. The isocline (\ref{IsoclineYgene}) is not defined when $A\,g(t)\geq -f_{\min}$. In these cases, the invariant set deduced from the isocline is the non-compact set $[f^{-1}(A),\,+\infty[$, which makes it impossible to apply the fixed point argument. To build a compact invariant set, we take $y>0$, which means that $f_{\min}+A\,g(t)<\frac{d\,y}{d\,t}<A\,g(t)$. As long as $y>0$, integration gives lower and upper solutions of $y$
$$
y_0+f_{\min}\,t+\int_0^tA\,g(\xi)\,d\xi<y(t)<y_0+\int_0^tA\,g(\xi)\,d\xi.
$$
From the $T$-periodicity of $g$, it follows that 
$$
y_0+f_{\min}\,T<y(T)<y_0.
$$
It is then only necessary to bound $y_0$ to establish that $y>0$ on $[0,\,T]$. Setting
$$
g_{\min}=\min_{t\in[0,\,T]}g(t)<0,
$$
the lower solution of $y$ gives 
$$
y_0-t\left(|f_{\min}|+A|g_{\min}|\right)\leq y(t),\qquad \forall t\in [0,\,T].
$$
Taking $y_0>T(|f_{\min}|+A|g_{\min}|)$ therefore gives $y>0$. The compact set
$$
K_0=\left[f^{-1}(A),\,2\,T\left(|f_{\min}|+A|g_{\min}|\right)\right]
$$
is therefore invariant under the Poincar\'e map, which completes the proof.
\qquad\end{proof}

%-------------------------------------------------------------------------------------

\subsection{Auxiliary solution}

Like in section \ref{SecPreRes}, the first derivative of $Y$ with respect to $A$ is introduced: $Z(t,\,A)=\frac{\partial\,Y}{\partial\,A}$. Applying the chain-rule to (\ref{ODE_Ygene}) shows that the $T$-periodic solution $Z$ satisfies
\begin{equation}
\frac{\textstyle d\,Z}{\textstyle d\,t}=f^{'}(Y)\,Z+g(t).
\label{ODE_Zgene}
\end{equation}

\begin{lemma}
In the limit of null forcing $A=0$, the mean values of $Z$ and $Z^2$ are
\begin{equation}
\overline{Z}(0)=0,\qquad \overline{Z^2}(0)=\frac{\textstyle 1}{\textstyle 2}\sum_{n=1}^{\infty}\frac{\textstyle a_n^2+b_n^2}{\textstyle 1+\left(n\,\omega\right)^2},
\label{Zeps0Gene}
\end{equation}
where $a_n$ and $b_n$ are the Fourier coefficients of the source (\ref{AdimGene}).
\label{LemmaZeps0Gene}
\end{lemma}

\begin{proof}
Null forcing gives $Y=0$, hence (\ref{ODE_Zgene}) becomes
$$
\frac{\textstyle d\,Z}{\textstyle d\,t}=-Z+g(t).
$$
This equation is integrated over $[0,\,T]$. $T$-periodicity of $Z$ and $\overline{g}=0$ yield $\overline{Z}(0)=0$. The solution $Z$ is therefore sought as a Fourier series with null mean value
$$
Z(t,\,0)=\sum_{n=1}^{\infty}\left(A_n\,\sin n\omega t+B_n\,\cos n\omega t\right).
$$
Injecting this series in the ODE satisfied by $Z(t,\,0)$ and using (\ref{AdimGene}) provides $$
A_n=\frac{\textstyle a_n+n\,\omega\,b_n}{\textstyle 1+\left(n\,\omega\right)^2},\qquad
B_n=\frac{\textstyle b_n-n\,\omega\,a_n}{\textstyle 1+\left(n\,\omega\right)^2}, \quad\mbox{ with } \quad A_n^2+B_n^2=\frac{\textstyle a_n^2+b_n^2}{\textstyle 1+\left(n\,\omega\right)^2}.
$$
Now, Parseval's formula leads to the last equality in (\ref{Zeps0Gene}).
%Elementary algebra lead to
%$$
%Z(t,\,0)=\sum_{n=1}^{\infty}\frac{\textstyle 1}{\textstyle \sqrt{1+\left(n\,\omega\right)^2}}\left(a_n\,\sin\left(n \omega t-\theta_n\right)+b_n\,\sin\left(n \omega t+\phi_n\right)\right),
%$$
%with $\theta_n=\arctan (n\,\omega)$ and $\phi_n=\arctan(1\,/\,(n\,\omega))$, which generalizes (\ref{Zeps0}). The square %of $Z$ may therefore be written as
%$$
%Z^2(t,\,0)=\sum_{n=1}^{\infty}\left(\frac{\textstyle 1}{\textstyle 1+\left(n\,\omega\right)^2}\,\Delta_n^2+\Theta_n \right),
%$$
%where
%$$
%\Delta_n=a_n\,\sin\left(n \omega t-\theta_n\right)+b_n\,\sin\left(n \omega t+\phi_n\right)
%$$
%and $\Theta_n$ is a product of trigonometric terms with different arguments. Consequently, the mean value of $\Theta_n$ %over $[0,\,T]$ is null. Straightforward trigonometric calculations yield
%$$
%\Delta_n^2=a_n^2\,\sin^2\left(n\omega t-\theta_n\right)+b_n^2\,\sin^2\left(n\omega t+\varphi_n\right)-a_n\,b_n\,\cos\left(2n\omega t-\theta_n+\varphi_n\right),
%$$
%hence
%$$
%\overline{\Delta_n^2}=\frac{\textstyle a_n^2+b_n^2}{\textstyle 2},
%$$
%which completes the proof.
\qquad\end{proof}

%-------------------------------------------------------------------------------------

\subsection{Mean dilatation of the crack}\label{SecGeneDilat}

The next result extends the theorem \ref{TheoremYbar}. It shows that a positive jump in the mean elastic displacement still occurs.

\begin{theorem}
The mean value of the $T$-periodic solution $Y$ in (\ref{ODE_Ygene}) is positive and increases strictly with the forcing amplitude:
$$
\overline{Y}>0,\qquad \frac{\textstyle \partial\,\overline{Y}}{\textstyle \partial\,A}>0.
$$
At small forcing levels, the following local estimate holds
\begin{equation}
\overline{Y}=f^{''}(0)\,\left(\frac{\textstyle A}{\textstyle 2}\right)^2\,\sum_{n=1}^{\infty}
\frac{\textstyle a_n^2+b_n^2}{\textstyle1+\left(n\,\omega\right)^2}+{\cal O}\left(A^3\right),
\label{YbarLocalGene}
\end{equation}
where $a_n$ and $b_n$ are the Fourier coefficients of the source (\ref{AdimGene}).
\label{TheoremYbarGene}
\end{theorem}

\begin{proof}
Based on the $T$-periodicity and the null mean value of $g$, the proof of theorem \ref{TheoremYbar} can be straightforwardly extended to prove $\overline{Y}>0$. To prove the second inequality, we proceed in three steps. First, as stated in the proof of proposition \ref{PropositionYgene}, null forcing gives $Y=0$, and hence $\overline{Y}(0)=0$. Secondly, lemma \ref{LemmaZeps0Gene} gives $\overline{Z}(0)=0$. Thirdly, the $T$-periodic function $W=\frac{\partial^2\,Y}{\partial\,A^2}=\frac{\partial\,Z}{\partial\,A}$ satisfies (\ref{ODE_W}). The sign of $f^{'}$ and $f^{''}$ in (\ref{ODE_Ygene}) show that the isocline of zero slope (\ref{IsoclineW}) is positive or null: $I_0^W(t)=0$ occurs only at points where $Z$ vanishes. The proof of proposition \ref{PropositionW} therefore still holds here, yielding $W(t)>0$ at all $t\in[0,\,T]$. We therefore obtain $\overline{W}(A)>0$, and hence $\overline{Z}(A)>\overline{Z}(0)=0$, which proves the second inequality.

Lastly, a Taylor expansion of $Y$ when $A \ll 1$ gives
$$
\begin{array}{lll}
\displaystyle
\overline{Y}(A) &=& \displaystyle \overline{Y}(0)+A\,\frac{\partial\,\overline{Y}}{\partial\,A}(0) +\frac{\textstyle A^2}{\textstyle 2}\,\frac{\partial^2\,\overline{Y}}{\partial\,A^2}(0)+{\cal O}(A^3),\\
[6pt]
&=&\displaystyle \frac{\textstyle A^2}{\textstyle 2}\,\overline{W}(0)+{\cal O}(A^3).
\end{array}
$$
Elementary calculations on (\ref{ODE_W}) lead to
$$
\overline{W}(0)=f^{''}(0)\,\overline{Z^2}(0), 
$$
where $\overline{Z^2}(0)$ is given in (\ref{Zeps0Gene}), which completes the proof.
\qquad\end{proof}

In the case of a pure sinusoidal source, $a_1=1$ and $a_{n>1}=b_n=0$: the estimate (\ref{YbarLocalGene}) recovers the estimate (\ref{YbarLocal}).

%\vspace{1.5cm} 

%-------------------------------------------------------------------------------------
%-------------------------------------------------------------------------------------

\section{Conclusion}

\subsection{Physical observables}

A set of nondimensionalized parameters was used throughout this study. To recover the physical observables, we recall (\ref{Adimentionalize}) and (\ref{AdimGene}). The results of physical interest are as follows.
\begin{itemize}
\item Theorems \ref{TheoremYbar} and \ref{TheoremYbarGene}: mean dilatation of the crack
\begin{equation}
\begin{array}{l}
\displaystyle
\overline{[u]}>0,\qquad \frac{\textstyle \partial\,\overline{[u]}}{\textstyle \partial\,\frac{v_0}{\beta\,d}}>0,\\
[12pt]
\displaystyle
\overline{[u]}=\left|{\cal F}^{''}(0)\right|\,\frac{\textstyle v_0^2}{\textstyle \beta^2\,d}\,
\sum_{n=1}^{\infty}\frac{\textstyle a_n^2+b_n^2}{\textstyle 1+\left(n\,\Omega\,/\,\beta\right)^2}
%\frac{\textstyle 1}{\textstyle 1+\displaystyle \left(\Omega/\beta\right)^2}
+{\cal O}\left(\frac{\textstyle v_0^3}{\textstyle \beta^3\,d^2}\right),
\end{array}
\label{Result1}
\end{equation}
where $a_n$ and $b_n$ are the Fourier coefficients of the normalized source. The main aim of the present study was to establish the validity of the inequalities in (\ref{Result1}) with any set of parameters. In the case of model 1 (\ref{Model1}) and a monochromatic source (\ref{ForcageSinus}), the local estimate obtained here matches the theoretical expression obtained in our previous study using a perturbation analysis method \cite{Lombard08}. In the local estimate of (\ref{Result1}), an improvement of one order of accuracy is reached with a purely monochromatic source: see (\ref{YbarLocal}). Lastly, note that the local estimate with model 1 is twice that obtained with model 2 (\ref{Model2}). 
\item Theorem \ref{TheoremYmax}: maximum aperture of the crack (monochromatic source only)
\begin{equation}
\begin{array}{l}
\displaystyle
\frac{\textstyle \partial\,[u]_{\max}}{\textstyle \partial\,\frac{v_0}{\beta\,d}}>0,\qquad \frac{\textstyle \partial^2\,[u]_{\max}}{\textstyle \partial\,\left(\frac{v_0}{\beta\,d}\right)^2}>0,\\
[10pt]
\displaystyle
\max\left(\frac{\textstyle 2\,v_0}{\textstyle \beta}\,\frac{\textstyle 1}{\textstyle \sqrt{1+\left(\Omega/\beta\right)^2}}, \frac{\textstyle 2\,v_0}{\textstyle \Omega}\,\frac{\textstyle 1+\Omega/\beta}{\textstyle \sqrt{1+\left(\Omega/\beta\right) ^2}}-\frac{\textstyle 2\,\pi\,\beta\,d}{\textstyle \Omega}\,|{\cal F}_{\max}|\right) \leq [u]_{\max} \leq \frac{\textstyle 4 \,v_0}{\textstyle \Omega},\\
%\frac{\textstyle 2\,v_0}{\textstyle \Omega}\,\frac{\textstyle 1}{\textstyle \sqrt{1+\left(\Omega/\beta\right)^2}}-\frac% {\textstyle 2\,\pi\,\beta\,d}{\textstyle \Omega}\,|{\cal F}_{\max}| \leq [u]_{\max} \leq \frac{\textstyle 4 \,v_0}%{\textstyle \Omega},\\
\\%[10pt]
\displaystyle
[u]_{\max}=\frac{\textstyle 2\,v_0}{\textstyle \beta}\,\frac{\textstyle 1}{\textstyle \sqrt{1+\left(\Omega/\beta\right)^2}}+{\cal O}\left(\frac{\textstyle v_0^2}{\textstyle \beta^2\,d}\right). 
\end{array}
\label{Result2}
\end{equation}
\item Theorem \ref{TheoremYmin}: maximum closure of the crack (monochromatic source only)

\begin{equation}
\begin{array}{l}
\displaystyle
-d < d\,\left(-{\cal F}\right)^{-1}\left(\frac{\textstyle 2\,v_0}{\textstyle \beta\,d}\right)\leq [u]_{\min} \leq 0,\\
[8pt]
\displaystyle
[u]_{\min}=-\frac{\textstyle 2\,v_0}{\textstyle \beta}\,\frac{\textstyle 1}{\textstyle \sqrt{1+\left(\Omega/\beta\right)^2}}+{\cal O}\left(\frac{\textstyle v_0^2}{\textstyle \beta^2\,d}\right),\\
[10pt]
\displaystyle
\lim_{\frac{v_0}{\beta\,d}\rightarrow+\infty}[u]_{\min}=-d. 
\end{array}
\label{Result3}
\end{equation}
\end{itemize}
It is worth noting that $\overline{[u]}>0$ and $[u]_{\min}$ bounded are purely nonlinear phenomena: with the classical linear law (\ref{JClin}), $\overline{[u]}=0$ and $[u]_{\min}$ is not bounded independently of $v_0$.

%-------------------------------------------------------------------------------------

\subsection{Acoustic determination of the contact law}

One of the applications of the present study is the characterization of the crack model, in particular the finite compressibility of the crack. This data is crucial in geomechanics and geohydrology, where it is linked to the transport of fluids across fractured rocks \cite{Adler99}. If $\overline{[u]}$, the source, and the physical parameters of $\Omega_0$ and $\Omega_1$ are known, then the second equation in (\ref{Result1}) provides a straightforward mean of determining $|{\cal F}^{''}(0)|/d$. Note that the stiffness $K$, which is involved in $\beta$ via the proposition \ref{PropositionODE}, is classically measured using acoustic methods \cite{Pyrak90}. It then suffices to measure $\overline{[u]}$. For this purpose, two possible methods come to mind.

The first method consists in measuring the dilatation of the crack mechanically. This requires installing two strain gauges around the crack \cite{Korshak02}. However, in many contexts such as those encountered in geosciences, this is not practicable. The second method consists in performing acoustic measurements of the diffracted elastic waves far from the crack. The proposition \ref{PropositionYperiod} can be used here, by writing the Fourier series of incident, reflected, and transmitted elastic displacements
\begin{equation}
\begin{array}{l}
\displaystyle
u_I(x,\,t)=\frac{\textstyle v_0}{\textstyle \Omega}\,\left\{\cos(\Omega\,t-k_0\,x)-1\right\},\\
[6pt]
\displaystyle
u_R(x,\,t)=R_0^a+\sum_{n=1}^\infty\left\{R_n^a\,\sin n\,(\Omega\,t+k_0\,x)+R_n^b\,\cos n\,(\Omega\,t+k_0\,x)\right\},\\
[6pt]
\displaystyle
u_T(x,\,t)=T_0^a+\sum_{n=1}^\infty\left\{T_n^a\,\sin n\,(\Omega\,t-k_1\,x)+T_n^b\,\cos n\,(\Omega\,t-k_1\,x)\right\},
\end{array}
\label{Fourier}
\end{equation}
where $k_0=\Omega/c_1$ and $k_1=\Omega/c_1$. Definition of $\overline{[u]}$ implies
\begin{equation}
\overline{[u]}=T_0^a-R_0^a+\frac{\textstyle v_0}{\textstyle \Omega}.
\label{ConcluEq1}
\end{equation}
On the other hand, the causality of the source and the continuity of the stress (\ref{JCsigma}) mean that
\begin{equation}
\rho_1\,c_1\,T_0^a+\rho_0\,c_0\,R_0^a=-\frac{\textstyle \rho_0\,c_0\,v_0}{\textstyle \Omega}.
\label{ConcluEq2}
\end{equation}
Solving (\ref{ConcluEq1}) and (\ref{ConcluEq2}) gives
\begin{equation}
\begin{array}{lll}
\overline{[u]}&=&\displaystyle -\frac{\textstyle \rho_0\,c_0+\rho_1\,c_1}{\textstyle \rho_1\,c_1}\,R_0^a+\frac{\textstyle \rho_1\,c_1-\rho_0\,c_0}{\textstyle \rho_1\,c_1}\,\frac{\textstyle v_0}{\textstyle \Omega},\\
[10pt]
&=&\displaystyle \frac{\textstyle \rho_0\,c_0+\rho_1\,c_1}{\textstyle \rho_0\,c_0}\,T_0^a+2\,\frac{\textstyle v_0}{\textstyle \Omega}.
\end{array}
\label{ConcluEq3}
\end{equation}
Purely acoustic measurements of $R_0^a$ or $T_0^a$ therefore give $\overline{[u]}$, from which the nonlinear crack parameter can be easily determined.

%-------------------------------------------------------------------------------------

\subsection{Future lines of investigation}

The monotonicity of ${\cal F}$ in (\ref{Fconcave}) - or equivalently $f$ in (\ref{Adimentionalize}) and (\ref{AdimGene}) - was the key ingredient used here to prove the properties of the crack: the strict increasing of ${\cal F}$ - or the strict  decreasing of $f$ - proves the uniqueness and attractivity of the periodic solution. In addition, the strict concavity of ${\cal F}$ - or the strict convexity of $f$ - shows that the mean value of the solution increases with the forcing parameter. Relaxing these hypotheses may lead to more complex situations, and require more sophisticated tools. One thinks for instance to the case of hysteretic models of interfaces \cite{Johnson85,Gusev03}. 

This paper was focused on the normal finite compressibility of a crack. Coupling with shear stress is also an important topic to examine \cite{Pecorari03}.

Lastly, further studies are also needed on networks of nonlinear cracks \cite{Richoux06}. These configurations frequently occur in applications. Direct numerical simulations have been carried out in \cite{Lombard08}, but a rigorous mathematical analysis is still required.

\section*{Acknowledgments}

We are grateful to Didier Ferrand, Martine Pithioux and Benjamin Ricaud for their useful comments on the manuscript.
Many thanks also to Jessica Blanc for her careful reading.

%-------------------------------------------------------------------------------------

\end{document}